\documentclass[12pt,draftclsnofoot,onecolumn]{IEEEtran}
\makeatletter
\newcommand\semihuge{\@setfontsize\semihuge{22.3}{22}}
\makeatother
\usepackage{algpseudocode}
\usepackage{algorithm}
\usepackage{algorithmicx}

\usepackage{lipsum} % For algorithm smape after or before:

\usepackage[dvips]{color}
\usepackage{comment}
\usepackage{todonotes}
\usepackage{epsf}
\usepackage{epsfig}
\usepackage{times}
\usepackage{epsfig}
\usepackage{graphicx}
\usepackage{bbold}
\usepackage{mathtools}
\usepackage{mathrsfs}
\usepackage{amssymb}
\usepackage{pdfpages}
\usepackage{epstopdf}
\usepackage{float}
\usepackage{bigints}
\usepackage{relsize}
\newfloat{algorithm}{t}{lop}
\usepackage{amsmath}
\usepackage{dsfont}
\usepackage{dsfont}
\usepackage{lettrine} % \lettrine[findent=1pt]{{{R}}}{}
\usepackage{amsmath,epsfig,amssymb,algorithm,algpseudocode,amsthm,cite,url}

\usepackage{subcaption}
\allowdisplaybreaks
\usepackage{csquotes}
\usepackage{color}
\usepackage{verbatim}
\usepackage[english]{babel}

\usepackage[justification=centering]{caption}
\usepackage{verbatim}

\newtheorem{theorem}{\bf Theorem}

\newtheorem{lemma}{\bf Lemma}

\renewcommand{\algorithmicrequire}{\textbf{Input:}}
\renewcommand{\algorithmicensure}{\textbf{Output:}}

\begin{document}
	
\title{  Beyond 5G with UAVs: Foundations of a 3D  Wireless Cellular Network \vspace{-0.1cm}}    

\author{\IEEEauthorblockN{  Mohammad Mozaffari$^1$, Ali Taleb Zadeh Kasgari$^2$, Walid Saad$^2$, Mehdi Bennis$^3$, and M\'erouane Debbah$^4$}\vspace{0.15cm}\\
	\IEEEauthorblockA{
		\small 	$^1$ Ericsson Research, Santa Clara, CA, USA, Email: \url{mohammad.mozaffari@ericsson.com}.\vspace{-0.00cm}\\
		$^2$ Wireless@VT, Electrical and Computer Engineering Department, Virginia Tech, VA, USA, Emails:\url{{alitk,  walids}@vt.edu}.\vspace{-0.00cm}\\
		$^3$ CWC - Centre for Wireless Communications, University of Oulu, Finland, Email: \url{bennis@ee.oulu.fi}.\vspace{-0.00cm}\\
		$^4$ Mathematical and Algorithmic Sciences Lab, Huawei France R \& D, Paris, France, and CentraleSup´elec,\vspace{-0.00cm}\\   Universit\`e Paris-Saclay, Gif-sur-Yvette, France, Email: \url{merouane.debbah@huawei.com}.
		\thanks{A preliminary conference version of this work appears in \cite{GlobecomVersion}.}
		\thanks{Mohammad Mozaffari joined Ericsson in July 2018. He was with Wireless@VT, Electrical and Computer Engineering Department, Virginia Tech, VA, USA, when this work was done.}
		%\thanks{This work was supported by the Army Research Office (ARO) under Grant W911NF-17-1-0593 and by the U.S. National Science Foundation under Grant CNS-1617896.} 
		%\thanks{\textcolor{black}{A preliminary conference version of this work appears in \cite{MozaffariConf}}.}
		%\thanks{This work was supported by the U.S. National Science
		%Foundation under Grants AST-1506297, by the Office of Naval Research (ONR) under Grant N00014-15-1-2709, and, by the ERC Starting
		%Grant 305123 MORE (Advanced Mathematical Tools for Complex Network
     	%Engineering), and by the Academy of Finland.}
	}\vspace{-0.42cm}}
\maketitle\vspace{-0.8cm}
\vspace{-0.1cm}
\begin{abstract}
In this paper, a novel concept of three-dimensional (3D) cellular networks, that integrate drone base stations (drone-BS) and cellular-connected drone users (drone-UEs), is introduced. For this new 3D cellular architecture, a novel framework for network planning for drone-BSs as well as latency-minimal cell association for drone-UEs is proposed. For network planning, a tractable method for drone-BSs' deployment based on the notion of truncated octahedron shapes is proposed that ensures full coverage for a given space with minimum number of drone-BSs. In addition, to characterize frequency planning in such 3D wireless networks, an analytical expression for the feasible integer frequency reuse factors is derived. Subsequently, an optimal 3D cell association scheme is developed for which the drone-UEs' latency, considering transmission, computation, and backhaul delays, is minimized. To this end, first, the spatial distribution of the drone-UEs is estimated using a kernel density estimation method, and the parameters of the estimator are obtained using a cross-validation method. Then, according to the spatial distribution of drone-UEs and the locations of drone-BSs, the latency-minimal 3D cell association for drone-UEs is derived by exploiting tools from optimal transport theory. %In particular, using optimal transport theory, the optimal 3D cell partitions are derived  according to the spatial distribution of drone-UEs and the drone-BSs' locations.
%	, the optimal cell association of drone-UEs is derived using optimal transport theory such that the latency for drone-UEs are minimized. 
Simulation results show that the proposed approach reduces the latency of drone-UEs compared to the classical cell association approach that uses a signal-to-interference-plus-noise ratio (SINR) criterion. In particular, the proposed approach yields a reduction of up to 46\% in the average latency compared to the SINR-based association. The results also show that the proposed latency-optimal cell association improves the spectral efficiency of a 3D wireless cellular network of drones.

\end{abstract} %\vspace{0.1cm}

\section{Introduction}%\vspace{-0.01cm}

Recent reports from the federal aviation administration (FAA)  shows that the number of unmanned aerial vehicles (UAVs), also known as drones, will exceed 7 million in 2020 \cite{FAAUAV}. Such a massive use of drones will have significant impacts on wireless networking. From a wireless  perspective, the two key roles of drones include: aerial base station (BS), and user equipment (UE) \cite{mozaffari2018tutorial,wu2018uav, bor}. Due to their flexibility and inherent ability for line-of-sight (LoS) communications, drone-BSs can provide broadband, wide-scale, and reliable wireless connectivity during  disasters and temporary
events \cite{mozaffari2, Kalantari, Zhang,  wu2018uav, wu2018common,HouraniModeling, bor, Ding1}. %In this regard, key examples of employing drone-BSs for wireless coverage and capacity enhancement are Google's Loon Project and Facebook's Aquila project  \cite{Sky}.
 Moreover, drone-BSs offer a promising solution for ultra-flexible deployment and cost-effective wireless services, without the prohibitive costs of terrestrial BSs. %Therefore, the use of drone-BSs can provide cost-effective wireless services to users. In particular, in rural areas where building complete infrastructure is not economically or practically feasible, the  

% On the other hand, UAVs will be also users of the wireless infrastructure. In such a use case, UAVs act as user equipments (UAV-UEs) that need to communicate with the ground BSs in order to achieve their mission goal (e.g., package delivery or deployment). Indeed, in order to support a largescale deployment of UAVs, a reliable wireless communication infrastructure is needed to effectively control the UAVs’ operations while sustaining the traffic stemming from their application services (e.g., package delivery, surveillance, remote sensing, or even virtual reality).

Meanwhile, drones can also act as UEs (i.e., cellular-connected drone-UEs) that must connect to a wireless network so as to operate. In particular, cellular-connected drone-UEs can be used for wide range of applications such as package delivery \cite{Anibal},
surveillance, remote sensing, and virtual reality.
  The key feature of drone-UEs is their ability to intelligently move in three dimensions and optimize their trajectory in order to efficiently complete their missions. Therefore, drone-UEs are widely used for delivery purposes such as in Amazon's prime air drone delivery service and drug delivery in medical applications.

Wireless networking with drones faces a number of challenges. For instance, for drone-BSs,  key design problems include 3D deployment and  network planning, performance analysis, resource allocation, and cell association.
For drone-UEs, there is a need for reliable and low latency communications  for efficient control. However, existing terrestrial cellular networks have been primarily designed for supporting ground users and are not able to readily serve aerial users. In fact, terrestrial BSs may not be able to meet the low-latency and reliable communication requirements of drone-UEs due to blockage effects and the specific design of the BSs' antennas which are not suitable for supporting users at high-elevation angles. Also, in areas with geographical constraints, terrestrial BSs may not be available to provide wireless service to drone-UEs. In such cases, the deployment of aerial drone-BSs is a promising opportunity for providing reliable wireless connectivity for drone-UEs.  Clearly, to support drones in wireless networking applications, there is a need for developing the novel concept of a \emph{3D cellular network} that incorporates both drone-BSs and drone-UEs. \vspace{-0.2cm}

\subsection{Related Works on Drone Communications}
Recent studies on drone communications have investigated various design challenges that include performance characterization, trajectory optimization, 3D deployment, user-to-drone association, and cellular-connected UAVs. For instance, in \cite{Kalantari}, the authors proposed an algorithm for jointly optimizing the locations and number of drones to maximize wireless coverage. The work in \cite{ALZ2} studied the optimal 3D
deployment of UAVs for maximizing the number of covered ground users under quality-of-service
(QoS) constraints. In \cite{FarajLetter}, the authors  proposed a framework for strategic placement of multiple drone-BSs that provides wireless connectivity for a large-scale ground network. However, the prior studies on  deployment of UAV base stations ignore the existence of flying drone-UEs.

In addition, the work in \cite{Letter_OT} presented a delay-optimal cell association scheme in a UAV-assisted terrestrial wireless network. The work in \cite{Vishal} studied the optimal user-UAV association for capacity improvement in UAV-enabled heterogeneous wireless networks. \textcolor{black}{The work in \cite{LyuOffloading} proposed a novel  hybrid network architecture for cellular systems by using UAVs as aerial base stations for data offloading. In particular, with the proposed framework in \cite{LyuOffloading},  the minimum throughput of mobile terminals is maximized by jointly optimizing the user partitioning,  the spectrum allocation, as well as the UAV trajectory.} In \cite{KalantariAssociation}, the authors proposed an algorithm for maximizing the sum-rate of ground users by joint optimization of user-to-drone-BSs association and wireless backhaul bandwidth allocation. The work in \cite{MozaffariFlightTime} proposed a novel cell association approach that maximizes the total data delivered to ground users by drone-BSs that have limited flight endurance. However, the previous works on user association in drone networks are limited to ground users and do not consider 3D aerial users. Moreover, the previous works do not analyze latency (due to e.g., communication, computation, and backhaul) which is a key metric in 3D drone communication systems.

While there exists a number of studies on cellular-connected drone-UEs \cite{Coexistence,challita2018cellular,zhang2017cellular}, the potential use of drone-BSs for serving drone-UEs has not been considered. For example,  in \cite{Coexistence}, the authors studied the coexistence of drone-UEs and ground users in cellular networks and characterized the downlink coverage performance. The work in \cite{challita2018cellular} proposed an interference-aware path planning approach for drone-UEs with the goal of minimizing their communication latency and their interference on terrestrial users. In \cite{azari2017reshaping}, the authors analyzed the downlink coverage performance of drone-UEs that communicate with terrestrial base stations. In \cite{zhang2017cellular}, the authors proposed a trajectory design scheme for minimizing the mission time of a single UAV-UE. Meanwhile, the authors in \cite{lyu2017blocking} characterized the performance of drone-UEs in uplink communications with ground BSs in terms of blocking probability and average achievable throughput. However, the existing studies on cellular-connected UAVs do not exploit the deployment of aerial base stations for enabling low-latency and reliable drone-UEs' communications.

However,  none of these previous works \cite{mozaffari2,wu2018uav, bor,Kalantari,KalantariAssociation, zhang2017cellular, MozaffariFlightTime,Coexistence,challita2018cellular,FarajLetter,Zhang,ALZ2,Letter_OT,Vishal, Naderi}, studied a 3D wireless network of co-existing aerial base stations and users (i.e., drone-BSs and drone-UEs) while addressing network planning, deployment, and latency-aware cell association problems. \vspace{-0.1cm}

\subsection{Contributions}
The main contribution of this paper is to introduce the novel concept of a fully-fledged drone-based 3D cellular network that incorporates drone-UEs, low-altitude platform (LAP)
drone-BSs, and high-altitude platform (HAP) drones. In this new 3D cellular network
architecture, we propose a framework for addressing the two fundamental problems of network planning and 3D cell association.  In particular, our proposed framework includes a tractable approach for three-dimensional placement and frequency planning for drone-BSs, as well as a latency-minimal 3D cell association scheme for servicing drone-UEs. For deployment, we introduce a new approach based on truncated octahedron cells that determines the minimum number of drone-BSs that can cover a 3D space, along with their locations. Furthermore, for frequency planning in the proposed 3D wireless network, we derive an analytical expression for the feasible integer frequency reuse factors. To perform latency-minimal 3D cell association, first, we estimate the spatial distribution of drone-UEs by using a kernel density estimation method. Then, given the locations of drone-BSs and the distribution of drone-UEs, we find the optimal 3D cell association for which the total latency of serving drone-UEs is minimized. In this case, we analytically characterize the optimal 3D cell partitions by exploiting tools from optimal transport theory. Our results show that the proposed approach significantly reduces the latency of serving drone-UEs, compared to classical cell association approach that uses signal-to-interference-plus-noise ratio (SINR) criterion. In particular, our approach yields around 46\% reduction in the average total latency compared to the SINR-bases association. The results also reveal that our latency-optimal cell association improves spectral efficiency of the considered 3D wireless network with drones.   

The rest of this paper is organized as follows. In Section II, we present the system model.
In Section III, the three-dimensional placement of drone-BSs is investigated. In Section IV, we describe our approach for estimating the spatial distribution of drone-UEs. Section V presents the proposed latency-optimal cell association scheme. Simulation results are provided in Section
VI and conclusions are drawn in Section VII.

\begin{figure}[!t]
	\begin{center}
		\vspace{-0.1cm}
		\includegraphics[width=12cm]{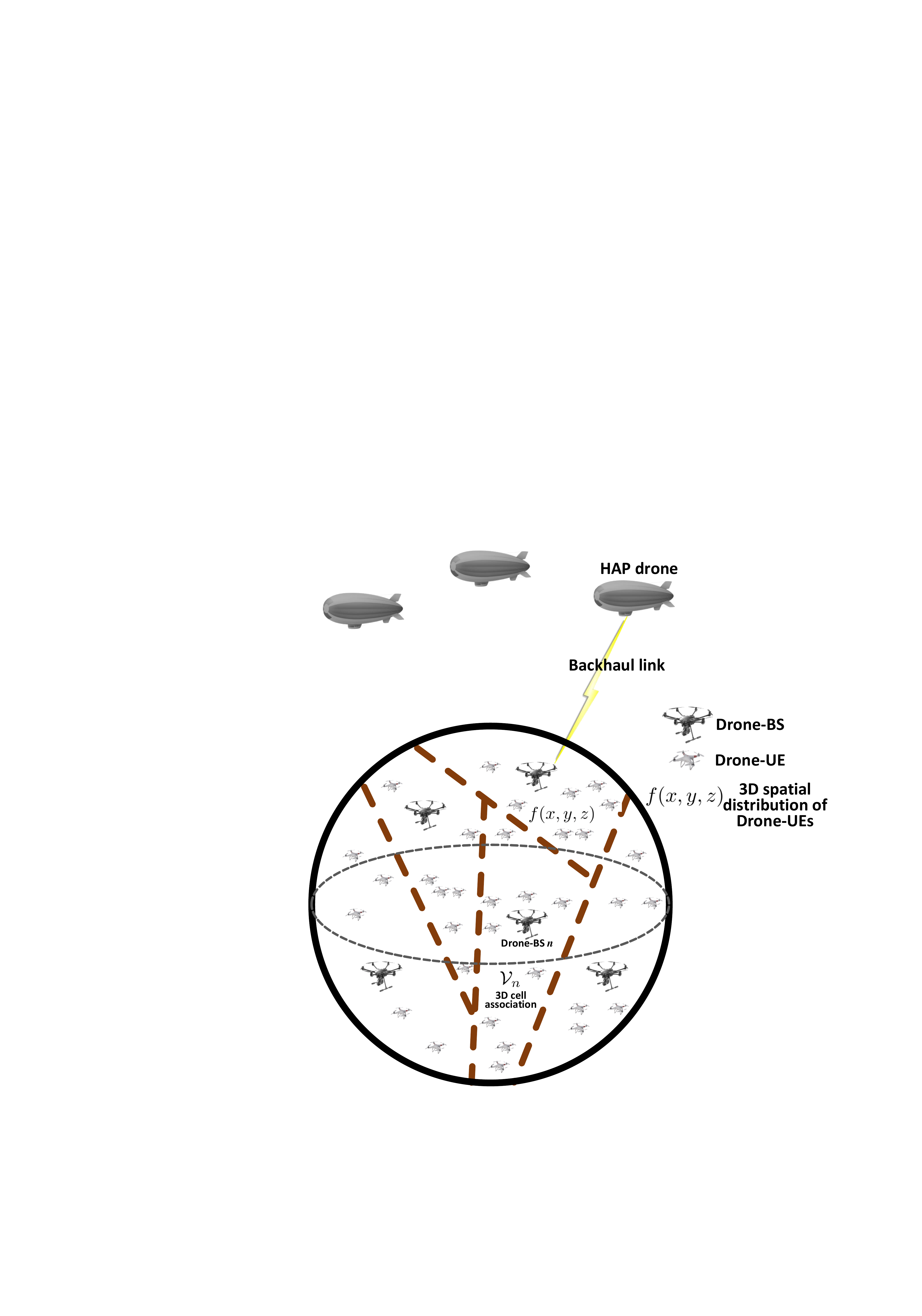}
		\vspace{-0.4cm}
		\caption{The proposed 3D wireless network with drone-BSs, drone-UEs, and HAP drones. \vspace{-.82cm}}
		\label{SystemModel}
	\end{center}	\vspace{-0.4cm}
\end{figure} 
 
\section{System Model}
Consider a 3D cellular network  composed of $L$ drone users,  $N$ LAP drone base stations, and a number of HAP drones, as shown in Fig.\,\ref{SystemModel}. We represent the sets of drone-UEs, and drone-BSs, respectively, by $\mathcal{L}$, and $\mathcal{N}$. \textcolor{black}{Here, we focus on a stand-alone aerial network that consists of flying drones.} 
In this aerial network, drone-BSs serve drone-UEs\footnote{\textcolor{black}{Note that, drone-BSs are an essential part of a 3D cellular network, since drone-UEs are not capable of continuously maintaining  LoS links and sending uplink data directly to HAPs due to their mobility and energy limitations.}} in the downlink, and \emph{HAP drones provide a wireless backhaul connectivity} \cite{horwath2007experimental} for drone-BSs. The key advantage of HAP drones is their ability to adjust their positions according to the locations of drone-BSs. In addition, due to their high altitudes, HAPs can establish LoS backhaul links to the drone-BSs. \textcolor{black}{Therefore, while it is possible to use various types of backhaul for the proposed 3D cellular network \cite{DaSilva2018ICC}, we used HAPs that can establish free space optical communications (FSO) backhaul links to the UAV-BSs due to the improved reliability and lower latency of  this link compared to a terrestrial BS backhaul.}  %Let ($x^{\textrm{HAP}}_m$, $y^{\textrm{HAP}}_m$, $z^{\textrm{HAP}}_m$) be the 3D Cartesian coordinate of pre-deployed HAP $m$.
  \textcolor{black}{In our proposed 3D cellular network, we adopt omni-directional antennas for drone-BSs to enable a full 3D connectivity.} Here, the deployment of drone-BSs is performed based on a 3D cellular structure which will be presented in Section III. % Also, we consider a frequency reuse factor of $q$ in our 3D cellular networking planning with drone-BSs.  
   For backhaul connectivity, we assume that each drone-BS connects to its closest HAP that can provide a maximum rate.  We denote the backhaul transmission rate for drone-BS $n$ by $C_n$, which is assumed to be given in our model\footnote{The backhaul rate can be easily calculated based on the locations of HAPs and drone-BSs, transmit power of HAPs, and bandwidth of backhaul links.}.  Drone-BS $n$ uses transmit power $P_n$ bandwidth $B_n$ in order to serve its associated flying drone-UEs. Let $f(x,y,z)$ be the spatial probability density function of drone-UEs which represents the probability that each drone-UE is present around a 3D location $(x,y,z)$.
   In our model, drone-BSs use machine learning tools to estimate the spatial probability distribution of drone-UEs, for a certain period of time, based on any available prior information about the drone-UEs' locations. By performing such estimation, the network will no longer need to continuously track the locations of flying drone-UEs thus alleviating the associated overhead.    %We estimate this spatial distribution function with $\hat f(x,y,z)$ as discussed in Section \ref{sec:EstDist}, and use  $\hat f(x,y,z)$ to determine the 3D cell association when serving drone-UEs. 
   To find the 3D cell association when serving drone-UEs, we partition the space into $N$ 3D cells each of which representing a volume that must be serviced by one drone-BS. Let $\mathcal{V}_n$ be a 3D space (i.e., 3D cell) associated to drone-BS $n$ that serves drone-UEs located within this cell. The average number of drone-UEs inside $\mathcal{V}_n$ is given by:
\begin{equation}
K_n=L\int_{\mathcal{V}_n}{f(x,y,z) \textrm{d}x\textrm{d}y\textrm{d}z}.
\end{equation}

We assume that each drone-BS adopts a frequency division multiple access (FDMA) technique (as done in \cite{MozaffariFlightTime} and \cite{FDMA2}) when servicing its associated drone-UEs. Hence, the average downlink transmission rate from a drone-BS $n$ to a drone-UE located at $(x,y,z)$ is:
\begin{equation}
R_n(x,y,z)=\frac{B_n}{K_n}\log_2\big(1+\gamma_n(x,y,z)\big),
\end{equation}
where $\frac{B_n}{K_n}$ is the amount of bandwidth for servicing each drone-UE located in $\mathcal{V}_n$, which is determined by sharing the total bandwidth among the drone-UEs. $\gamma_n(x,y,z)$ is the SINR for a drone-UE located at $(x,y,z)$ served by drone-BS $n$. 

We consider the \emph{average latency} in servicing drone-UEs as our main performance metric. In particular, we consider transmission latency in drone-BSs to drone-UEs communications, backhaul latency in drone-BSs to HAP drones links, and computation latency for drone-BSs that serve drone-UEs. The transmission latency for a drone-UE located at $(x,y,z)$ which is served by drone-BS $n$ is\,\footnote{\textcolor{black}{(\ref{TransLat}) represents the transmission latency which is the time required to transmit an entire packet \cite{Hernandez2007Delay}.}}:
\begin{equation} \label{TransLat}
\tau^\textrm{Tr}_n(x,y,z,K_n)=\frac{\beta}{R_n(x,y,z)},
\end{equation}
where $\beta$ is the number of bits per packet that must be transmitted to each drone-UE.

The backhaul latency depends on the load of drone-BSs and the backhaul transmission rates. In this case, the average backhaul latency in drone-BS $n$ to its corresponding HAP-drone communications is given by:
\begin{equation}
\tau^\textrm{B}_n(K_n)=\frac{\beta L \displaystyle\int_{\mathcal{V}_n}{f(x,y,z) \textrm{d}x\textrm{d}y\textrm{d}z}}{C_n}=\frac{\beta K_n}{C_n},
\end{equation}
where $C_n$ is the maximum backhaul transmission rate for drone-BS $n$, and $\beta L\int_{\mathcal{V}_n}{f(x,y,z) \textrm{d}x\textrm{d}y\textrm{d}z}$ represents the average load on drone-BS $n$.

The computation time depends on the data size (i.e., load) that must be processed in each drone-BS, and the processing speed. To model the computational latency at drone-BS $n$, we use function $g_n(\beta K_n)$ with $\beta K_n$ being the total data size that must be processed at the drone-BS. Therefore, the total latency experienced by any arbitrary  drone-UE located at $(x,y,z)$ while being served by drone-BS $n$ can be given by:
\begin{equation}
\tau^\textrm{tot}_n(x,y,z,K_n)=\tau^\textrm{Tr}_n(x,y,z,K_n)+\tau^\textrm{B}_n(K_n)+g_n(\beta K_n),
\end{equation}

 Given this model, our goal is to minimize the average latency of drone-UEs by finding the optimal 3D cell association in drone-BSs to drone-UEs communications. In particular, given the locations of drone-BSs which are deployed based on a 3D cellular structure (in Section III), and the estimated spatial distribution of drone-UEs (in Section IV), we determine the optimal 3D cell partitions $\mathcal{V}_n$, $\forall n\in \mathcal{N}$ that lead to a minimum average latency for drone-UEs. In this regard, our 3D cell association optimization problem can be posed as follows:
  \begin{align} \label{OPT1}
  \mathop {\min }\limits_{{\mathcal{V}_1,...,\mathcal{V}_N}} &\mathlarger{\sum}\limits_{n = 1}^N \Bigg[ {\bigintssss_{{\mathcal{V}_n}}{\tau^\textrm{Tr}_n\big(x,y,z, K_n\big) f(x,y,z)}} \textrm{d}x\textrm{d}y\textrm{d}z+\tau^\textrm{B}_n(K_n)+g_n(\beta K_n)\Bigg], \\
  \textrm{s.t.}\,\,
  &{\mathcal{V}_l} \cap {\mathcal{V}_m} = \emptyset ,\,\,\,\forall l \ne m \in \mathcal{N}, \label{area1}\\
  &\bigcup\limits_{n \in \mathcal{N}} {{\mathcal{V}_n}}  = \mathcal{V}, \label{area2}
  \end{align}
 where  $K_n=L\int_{\mathcal{V}_n}{f(x,y,z) \textrm{d}x\textrm{d}y\textrm{d}z}$ is the average number of drone-UEs in $\mathcal{V}_n$ which depends on the 3D cell association, and $\mathcal{V}$ is the entire considered space in which drone-UEs can fly. The constraints in  (\ref{area1}) and (\ref{area2}) ensure that the 3D association spaces are disjoint and their union covers the considered space $\mathcal{V}$. {\color {black}
 	Table \ref{TablePara} provides a list of our main parameters and notations.}

 \begin{table}[!t]
 	\normalsize
 	\begin{center}{\color {black}
 		%\centering
 		\caption{\small List of notations.}
 		\vspace{-0.1cm}
 		\label{TablePara}
 		\resizebox{12cm}{!}{
 			\begin{tabular}{|c|c|}
 				\hline
 				\textbf{Notation} & \textbf{Description} \\ \hline \hline
 				$f_c$	&     Carrier frequency         \\ \hline 
 				$P_n$	&    Drone-BS transmit power      \\ \hline
 				
 				$N_o$	&     Noise power spectral density    \\ \hline
 				
 				$L$	&     Number of drone-UEs     \\ \hline

 				$B_n$	&    Bandwidth for each drone-BS      \\ \hline
 				
 				$\alpha$	&    Path loss exponent      \\ \hline
 				
 				$\eta$	&     Path loss constant     \\ \hline
 				
 				$\beta$	&     Packet size for drone-UE      \\ \hline
 				
 				$q$	&     Frequency reuse factor      \\ \hline
 				
 				$C_n$	&     Backhaul rate for drone-BS $n$    \\ \hline

 				$f(x,y,z)$	&     Spatial distribution of drone-UEs      \\ \hline

 				$\hat{f}(x,y,z)$	&     Estimated spatial distribution of drone-UEs      \\ \hline
 				
 				$L$	&     Number of drone-UEs      \\ \hline

 				$\mathcal{V}_n$	&     3D cell partition associated with drone-BS $n$    \\ \hline
 				
 				$K_n$	&     Average number of drone-UEs inside $\mathcal{V}_n$    \\ \hline
 				
 				$R_n(x,y,z)$	&     average transmission rate from drone-BS $n$ to a drone-UE  located at $(x,y,z)$     \\ \hline
 				
 				$\tau^\textrm{Tr}_n$	&     Transmission latency for drone-BS $n$    \\ \hline
 				
 				$\tau^\textrm{B}_n$	&     Backhaul latency for drone-BS $n$      \\ \hline
 				
 				$\tau^\textrm{tot}_n(x,y,z,K_n)$	&     Total latency experienced by a drone-UE located at $(x,y,z)$ served by drone-BS $n$     \\ \hline
 				
 				$R$	&     Edge length of a truncated octahedron      \\ \hline
 				
 				$\gamma_n(x,y,z)$	&     SINR for a drone-UE located at $(x,y,z)$ served by drone-BS $n$      \\ \hline
 				
 				$g_n$	&     Computational latency at drone-BS $n$   \\ \hline
 				
 				$\omega_n$	&    Computation constant (i.e., speed) for each drone-BS     \\ \hline

 				$\mu_x, \mu_y, \mu_z$& Mean of the truncated Gaussian distribution in $x$, $y$, and $z$ directions\\ \hline
 				
 				$\sigma_x, \sigma_y, \sigma_z$& Standard deviation of the distribution in $x$, $y$, and $z$ directions \\ \hline
 		\end{tabular}}
 	}
 	\end{center}%\vspace{-0.5cm}
 \end{table}

% We will study, in details, our 3D cell association problem in Section V. 

In Fig.\,\ref{Diagram}, we summarize the key steps for developing our proposed drone-based 3D cellular network architecture. First, we plan the network deployment of drone-BSs based on a truncated octahedron scheme that can ensure full coverage with a minimum number of drone-BSs. Second, using some available information about the locations' history of drone-UEs, we estimate the 3D spatial distribution of the drone-UEs for a given period of time. Finally, given the locations of drone-BSs and the spatial distribution of drone-UEs, we derive an optimal 3D cell association rule for which the latency of servicing drone-UEs is minimized. \textcolor{black}{Note that, we consider a relatively long-term deployment of drones which can be updated after a specific period of time, if needed. For each deployment configuration, one needs to optimally perform cell association based on the distribution of drone-UEs so as to enhance the system performance.}

\begin{figure}[!t]
	\begin{center}
		\vspace{-0.1cm}
		\includegraphics[width=14cm]{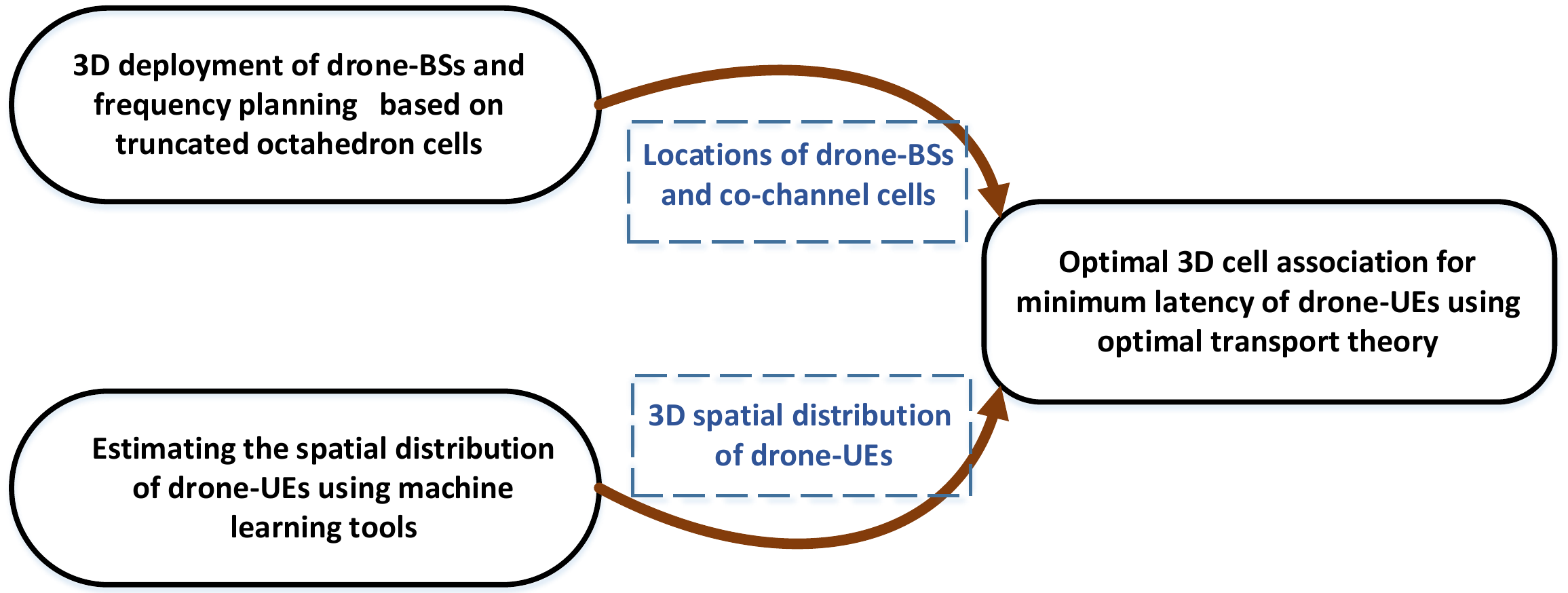}
		\vspace{-0.4cm}
		\caption{ Our proposed framework for designing the 3D cellular network. \vspace{-.82cm}}
		\label{Diagram}
	\end{center}	\vspace{-0.1cm}
\end{figure}

 \section{Three-dimensional Network Planning of Drone-BSs: A Truncated Octahedron Structure} \label{Sec}
To perform 3D network planning, we propose a framework for the 3D deployment of drone-BSs and associated frequency planning. In particular, we use the notion of truncated octahedron structure to determine the drone-BSs' locations as well as the feasible integer frequency factors that allow finding co-channel interfering drone-BSs. 
 
 In traditional ground cellular networks, hexagonal cell shapes are used while deploying base stations. This is due to the fact that, a 2D space can be fully covered (i.e., without any gaps) by non-overlapping hexagons. While triangle and square cells are also able to tessellate a 2D area, the hexagonal cell is preferred in cellular wireless network planning due to the following reasons. First, the hexagonal shape has a larger area than the square and the triangle, hence less cells will be needed to cover a geographical area. Second,   the hexagonal cell reasonably  approximates the circular radiation pattern of an omni-directional antenna base station. 
 
 Inspired by 2D cellular networks, we propose a framework for the deployment of a 3D cellular network. In three dimensions, the regular polyhedron geometric shapes that can tessellate the space  (i.e., fill the 3D space entirely) include cube, hexagonal prism, rhombic dodecahedron, and truncated octahedron \cite{alam2006coverage}. Among these 3D shapes, the truncated octahedron is the closest approximation of a sphere. Moreover, the number of polyhedron required for completely covering a 3D space is minimized by adopting the truncated octahedron \cite{alam2006coverage}. Therefore, in our model, we use the truncated octahedron structure for deploying the drone-BSs.

The truncated octahedron is a polyhedron in three dimensions with regular polygons faces. As we can see from Fig.\,\ref{Octahedron}, the truncated octahedron has 14 faces with 8 regular hexagonal and 6 square, 24 vertices, and 36 edges \cite{coxeter1973regular}. The key feature of the truncated octahedron is that it can tessellate the three-dimensional Euclidean space. In other words, the 3D space can be completely filled with multiple copies of the truncated octahedron without any overlap. We exploit this feature of the truncated octahedron in our 3D cellular network deployment with drone-BSs. 

%In traditional ground cellular networks, hexagonal cell shapes are used while deploying base stations. This is due to the fact that, a 2D space can be fully covered (i.e., without any gaps) by non-overlapping hexagons. While triangle and square cells are also able to tessellate a 2D area, the hexagonal cell is preferred in cellular wireless network planning due to the following reasons. First, the hexagonal shape has a larger area than the square and the triangle, hence less cells will be needed to cover a geographical area. Second,   the hexagonal cell reasonably  approximates the circular radiation pattern of an omni-directional antenna base station. 
%
%Inspired by 2D cellular networks, we propose a framework for the deployment of a 3D cellular network. In three-dimensions, the regular polyhedron geometric shapes that can tessellate the space  (i.e., fill the 3D space entirely) include cube, hexagonal prism, rhombic dodecahedron, and truncated octahedron \cite{alam2006coverage}. Among these 3D shapes, the truncated octahedron is the closest approximation of a sphere. Moreover, the number of polyhedron required for completely covering a 3D space is minimized by adopting the truncated octahedron \cite{alam2006coverage}. Therefore, in our model, we use the truncated octahedron structure for deploying the drone-BSs.  

 \begin{figure}[!t]
 	\begin{center}
 		\vspace{-0.1cm}
 		\includegraphics[width=4.2cm]{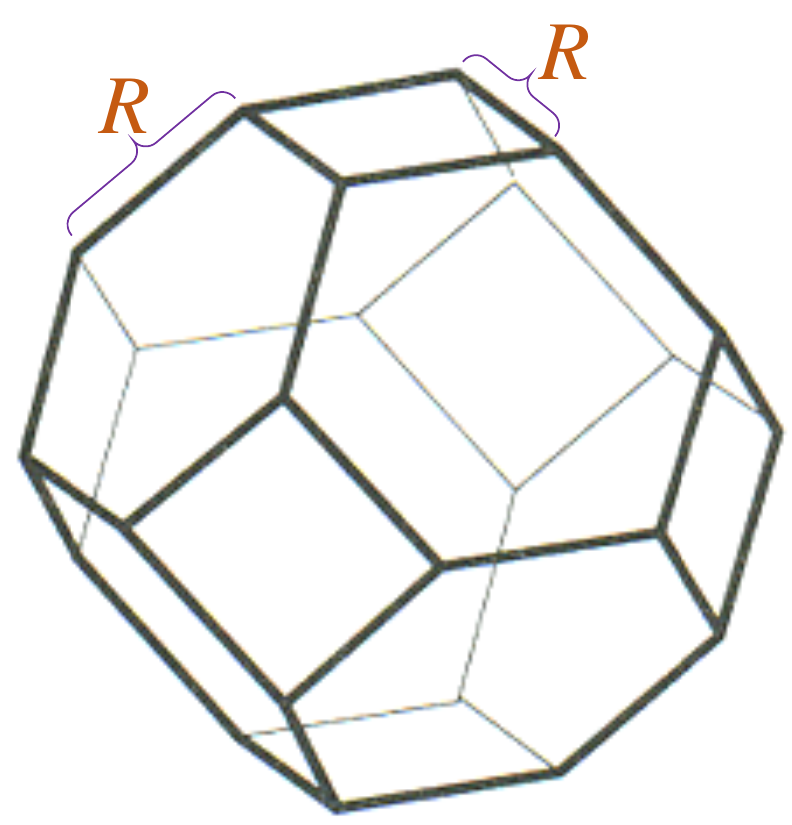}
 		\vspace{-0.4cm}
 		\caption{ Truncated octahedron in 3D. \vspace{-.82cm}}
 		\label{Octahedron}
 	\end{center}	\vspace{-0.1cm}
 \end{figure}

\begin{figure}[!t]
	\begin{center}
		\vspace{-0.1cm}
		\includegraphics[width=6.5cm]{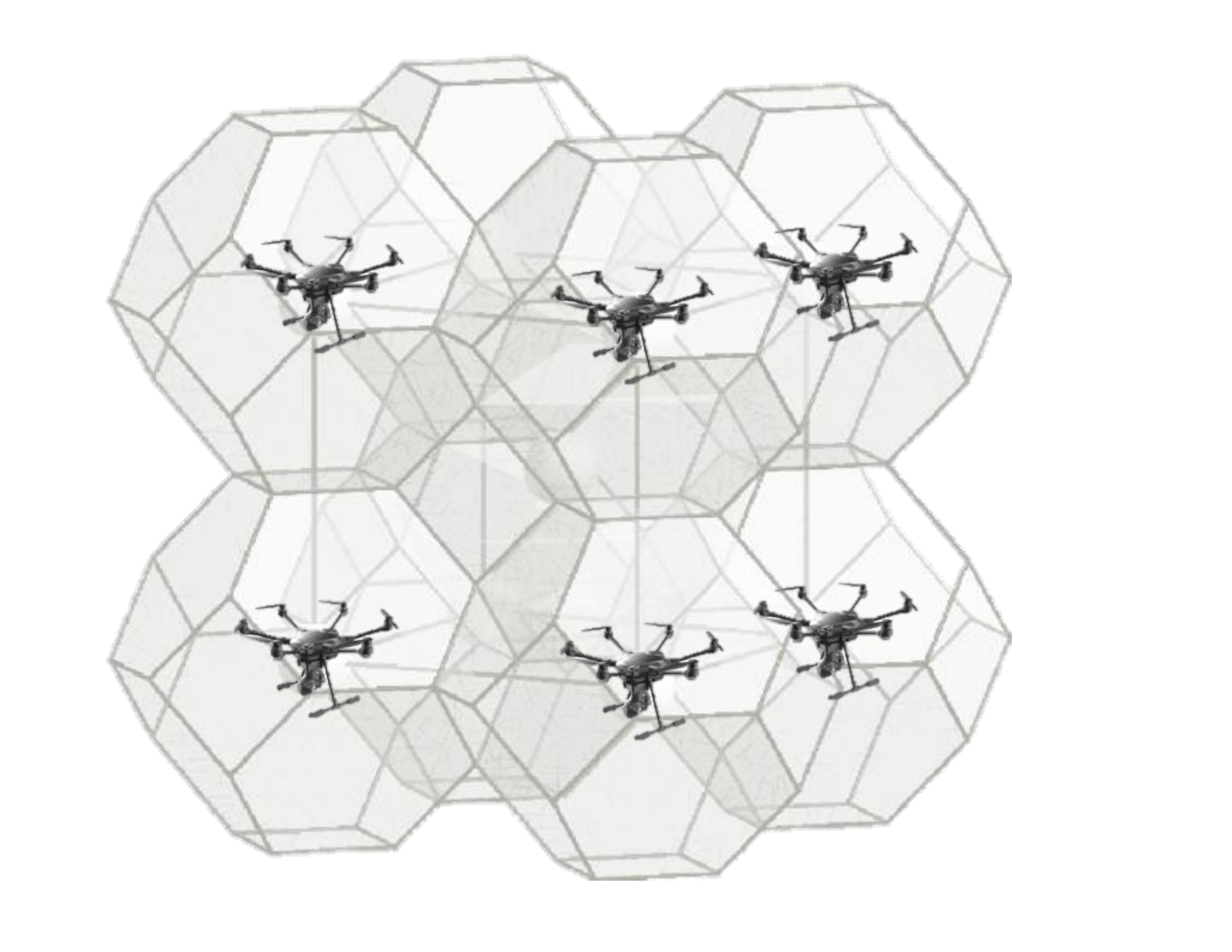}
		\vspace{-0.4cm}
		\caption{ Deployment of drone-BSs based on truncated octahedron cells. \vspace{-.82cm}}
		\label{3DCells}
	\end{center}	\vspace{-0.1cm}
\end{figure}

The deployment of drone-BSs needs to be done such that the entire desired space is covered. To this end, we first completely fill the given space with an arrangement of multiple truncated octahedron cells. Then,  we place each drone-BS at the center of each truncated octahedron, as shown in Fig.\,\ref{3DCells} as an illustrative example. Our proposed deployment approach can ensure full coverage for a given 3D space and is also easy to implement and tractable. Moreover, our approach facilitates frequency planning in 3D cellular networks by deriving analytical expressions for the feasible integer reuse factors. Next, we determine the locations of drone-BSs based on the proposed truncated octahedron cell structure.

 \begin{theorem} \label{T1} \normalfont
 The three-dimensional locations of  drone-BSs in the proposed 3D cellular network are given by:
 \begin{equation} \label{3DLocation}
\boldsymbol{P}_{\{a,b,c\}}=\big[x_o,y_o,z_o\big]+\sqrt{2}R\Big[a+b-c,-a+b+c,a-b+c\Big],
 \end{equation}	
 where $a$, $b$, $c$ are integers chosen from set $\{...,-2,-1,0,1,2,...\}$, and $R$ is the edge length of the considered truncated octahedrons. $[x_o,y_o,z_o]$ is the Cartesian coordinates of a given reference location (e.g., center of a specified space).
 \end{theorem}
\begin{proof}
For the deployment of drone-BSs, we first create a  3D lattice of truncated octahedrons and then, place each drone-BS at the center of each truncated octahedron. Hence, to determine the locations of drone-BSs, we need to find the center of truncated octahedrons.

 \begin{figure}[!t]
 	\begin{center}
 		\vspace{-0.1cm}
 		\includegraphics[width=6cm]{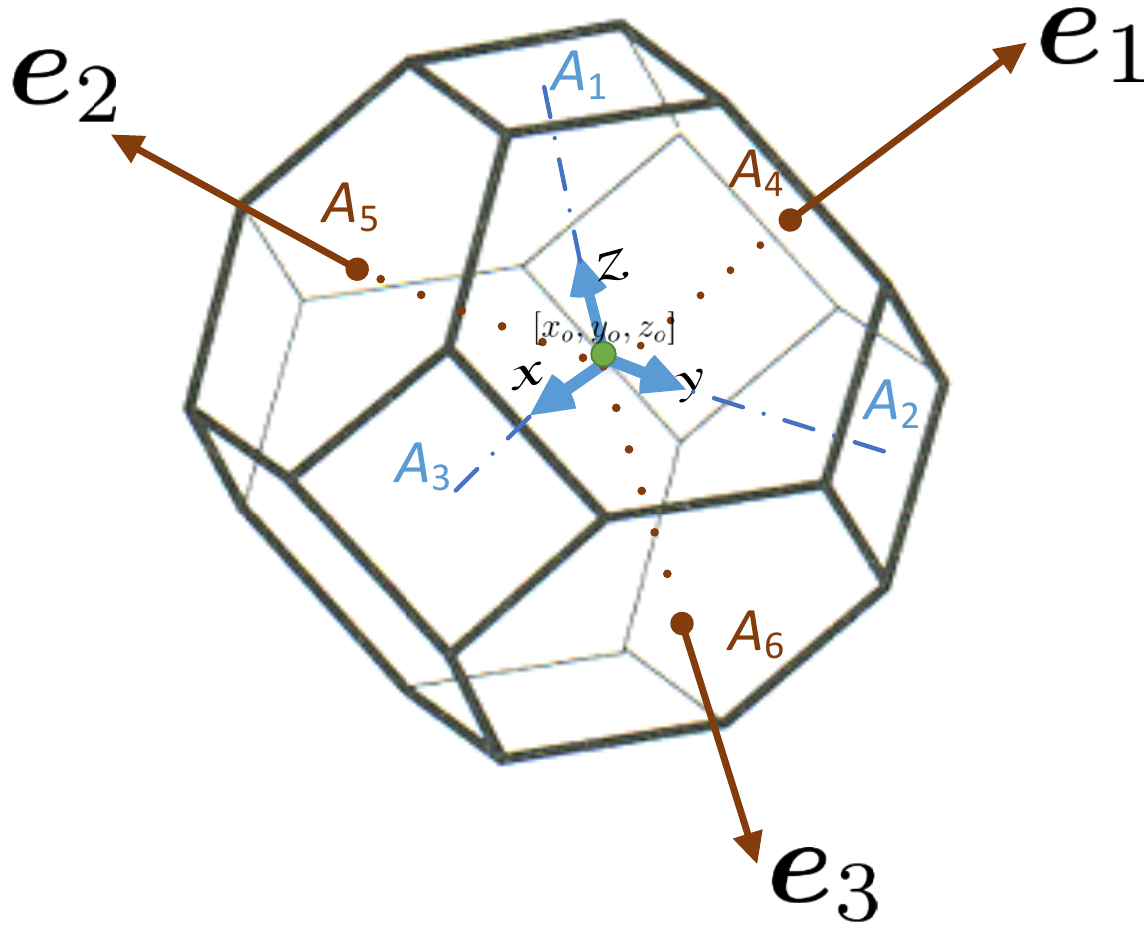}
 		\vspace{-0.4cm}
 		\caption{ Coordinate systems in drone-BSs deployment. \vspace{-.82cm}}
 		\label{Theorem1}
 	\end{center}	\vspace{-0.1cm}
 \end{figure} 

Let $[x_o,y_o,z_o]$ be the center of the first truncated octahedron in  Cartesian coordinates with the $x$, $y$, and $z$ directions being perpendicular to square faces $A_3$, $A_2$, and $A_1$ as shown in Fig.\, \ref{Theorem1}. We find a new coordinate system whose integer coordinates are the center of the truncated octahedrons. By moving, in integer value steps, along the axes of this coordinate system, we can reach the center of the truncated octahedrons. We consider a coordinate system whose axes ($\boldsymbol{e}_1$, $\boldsymbol{e}_2$, $\boldsymbol{e}_3$) are vertically outward the hexagonal faces, $A_4$, $A_5$, and $A_6$. Now, we find the Euclidean length of each unit axis of this coordinate system. The distance between the center of the truncated octahedron (e.g., $[x_o,y_o,z_o]$) and each hexagonal face is $R\sqrt{6}/2$ \cite{coxeter1973regular}. Therefore, the distance from $[x_o,y_o,z_o]$ to the center of an adjacent truncated octahedron connecting to face $A_4$ is $R\sqrt{6}$. As a result, each unit on axis $\boldsymbol{e}_1$ (and also $\boldsymbol{e}_2$ and $\boldsymbol{e}_3$) must be $2R\sqrt{6}$. It can be easily verified that the centers of the truncated octahedrons in the 3D lattice are the integer coordinates of the ($\boldsymbol{e}_1$,$\boldsymbol{e}_2$,$\boldsymbol{e}_3$) coordinate system. Hence, the 3D location of each drone-BS can be represented by a triple $(a,b,c)$ with $a$, $b$, and $c$ being integers. The position of a drone-BS obtained by $\{a,b,c\}$ is given by:
\begin{equation}\label{LOC}
\boldsymbol{P}_{\{a,b,c\}}=a \boldsymbol{e}_1+ b \boldsymbol{e}_2+ c \boldsymbol{e}_3.
\end{equation}  
Now, we need to represent $\boldsymbol{P}_{\{a,b,c\}}$ using Cartesian coordinates. To this end, we find the projection of $\boldsymbol{e}_1$, $\boldsymbol{e}_2$, and $\boldsymbol{e}_3$ on the $x$, $y$, and $z$ axes. With some geometric calculations and using the fact that the dihedral angle (i.e., angle between two intersecting planes) between the adjacent square face and hexagonal face is $\cos^{-1}(\frac{-1}{\sqrt{3}})$ \cite{coxeter1973regular},  we obtain: 
\begin{equation} \label{Stage1}
\begin{cases}
\boldsymbol{e}_1=R\sqrt{6}\big(\frac{-1}{\sqrt{3}}\boldsymbol{x}+\frac{1}{\sqrt{3}}\boldsymbol{y}+\frac{1}{\sqrt{3}}\boldsymbol{z}\big),\\
\boldsymbol{e}_2=R\sqrt{6}\big(\frac{1}{\sqrt{3}}\boldsymbol{x}+\frac{-1}{\sqrt{3}}\boldsymbol{y}+\frac{1}{\sqrt{3}}\boldsymbol{z}\big),\\
\boldsymbol{e}_3=R\sqrt{6}\big(\frac{1}{\sqrt{3}}\boldsymbol{x}+\frac{1}{\sqrt{3}}\boldsymbol{y}+\frac{-1}{\sqrt{3}}\boldsymbol{z}\big).
\end{cases}
\end{equation}
Finally, using (\ref{LOC}) and (\ref{Stage1}), the 3D locations of drone-BSs in Cartesian coordinates, with respect to the reference position $[x_o,y_o,z_o]$ are given by:
\begin{equation}\label{POS}
\boldsymbol{P}_{\{a,b,c\}}=\big[x_o,y_o,z_o\big]+\sqrt{2}R\Big[a+b-c,-a+b+c,a-b+c\Big],
\end{equation} 
which proves the theorem.
\end{proof}
Using Theorem \ref{T1}, we can find the 3D coordinates of drone-BSs which are deployed at the centers of truncated octahedrons. Moreover, as shown next, Theorem \ref{T1} allows us to determine the frequency reuse factor as well as interfering drone-BSs in the proposed 3D cellular network. 	
	
\begin{theorem} \label{Th2}\normalfont
	In the considered 3D cellular network, any feasible integer frequency reuse factors can be determined by solving the following equations:  
\begin{equation} \label{Theorem2}
\begin{cases}
\displaystyle q=\sqrt{\frac{\big[3(n_1^2+n_2^2+n_3^2)-2(n_1n_2+n_1n_3+n_2n_3)\big]^3}{27}},\vspace{0.2cm}\\
\displaystyle q=\sqrt{\frac{\big[3(m_1^2+m_2^2+m_3^2)-2(m_1m_2+m_1m_3+m_2m_3)\big]^3}{64}},
\end{cases}
\end{equation}
where $q$ is the frequency reuse factor which is a positive integer. $n_1$, $n_2$, $n_3$, $m_1$, $m_2$, and $m_3$ are integers that satisfy (\ref{Theorem2}) by generating feasible frequency reuse factors. 
\end{theorem}
\begin{proof}
We consider a truncated octahedron cell with 14 faces, as a reference cell. In this case, the number of first tier co-channel interfering cells is 14. Since the distance between centers of the reference cell and its adjacent cell is varies depending on the connecting face (i.e., hexagonal or square face), we consider two different co-channel distances (i.e., reuse distances). Let $D_u$ and $D_l$ be two different reuse distances to different interfering cells.

Assume that the center of a co-channel cell at a distance $D_l$ is located at a positive integer coordinate $(n_1,n_2,n_3)$ in our defined coordinate system ($\boldsymbol{e}_1, \boldsymbol{e}_2,\boldsymbol{e}_3$). Now, using (\ref{3DLocation}) in Theorem \ref{T1} leads to:
\begin{align} \label{Dl}
D_l&=\sqrt{2}R\sqrt{(n_1+n_2-n_3)^2+(-n_1+n_2+n_3)^2+(n_1-n_2+n_3)^2}\nonumber\\
&\overset{(a)}{=}R\sqrt{6(n_1^2+n_2^2+n_3^2)-4(n_1n_2+n_1n_3+n_2n_3)},
\end{align}
where in $(a)$ we used algebraic identities. 

 \begin{figure}[!t]
 	\begin{center}
 		\vspace{-0.1cm}
 		\includegraphics[width=6.7cm]{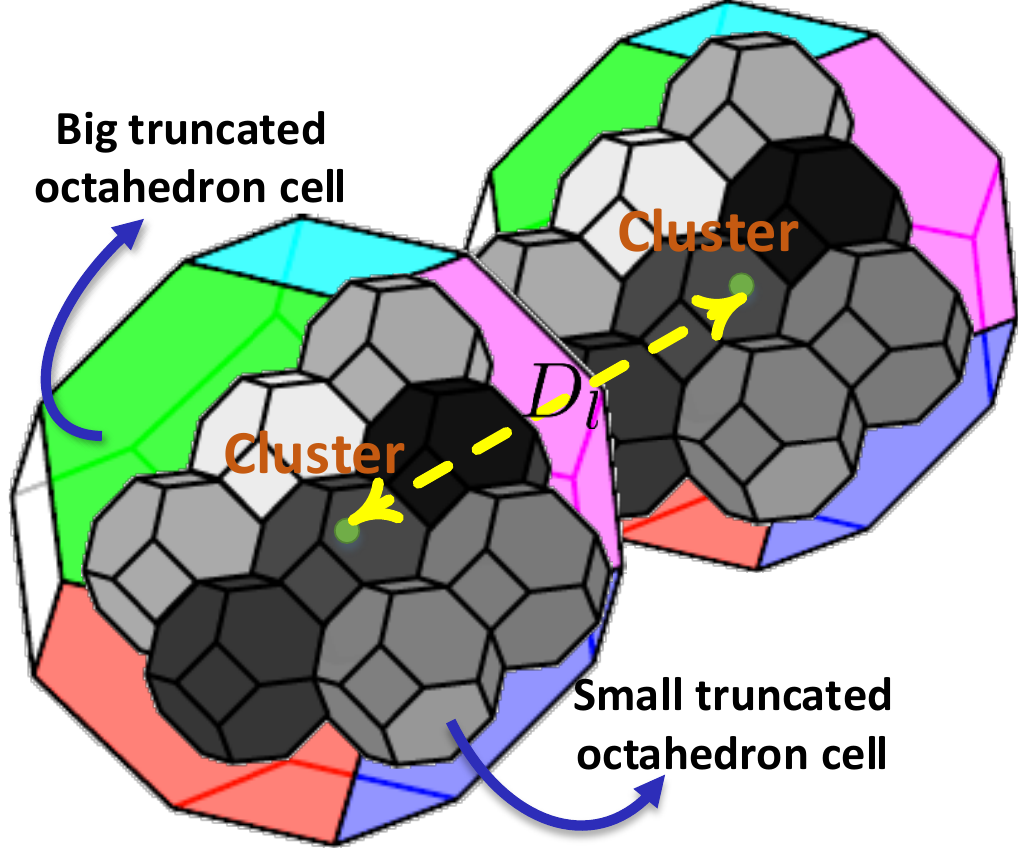}
 		\vspace{-0.4cm}
 		\caption{ Clusters of truncated octahedron cells. \vspace{-.82cm}}
 		\label{T2}
 	\end{center}	\vspace{-0.2cm}
 \end{figure}

Similar to 2D cellular networks, the frequency reuse factor is equal to the number of non-interfering cells within a cluster of cells. Hence, cells within each cluster will use different frequency bands. To find the frequency reuse factor in the 3D network, we compute the number of non-interfering cells that create one 3D cluster. Clearly, for the reference cell, a co-channel interfering cell is located in the adjacent cluster. Here, a given space is covered by multiple clusters of truncated octahedron cells. In addition, any space can be fully covered by a number of arbitrary-sized truncated octahedrons. Therefore, we can replace each cluster of cells with a big truncated octahedron cell (as illustrated in Fig.\,\ref{T2}) of the same volume. In this case, the centers of two co-channel cells are also the centers of two adjacent big truncated octahedron cells, as shown in Fig.\,\ref{T2}. These two big cells can be connected to each other either from  their hexagonal face (reuse distance $D_l$) or square face (reuse distance $D_u$). For the hexagonal case, the edge length of the big cells, $R_B$, is related to the reuse distance by:
\begin{equation}
R_B=\frac{D_l}{\sqrt{6}}.
\end{equation} 

The number of cells per cluster is equivalent to the volume ratio of the big cell (i.e., cluster) to one truncated octahedron cell:
\begin{equation}\label{Nnumber}
q=\frac{V_B}{V_S}\overset{(a)}{=}\frac{8\sqrt{2}R_B^3}{8\sqrt{2}R^3}=\big(\frac{D_l}{\sqrt{6}R}\big)^3\overset{(b)}{=}\sqrt{\frac{\big[3(n_1^2+n_2^2+n_3^2)-2(n_1n_2+n_1n_3+n_2n_3)\big]^3}{27}},
\end{equation}
where $V_B$ and $V_S$ are, respectively, the volumes of one cluster (e.g., big truncated octahedron) and a truncated octahedron cell.  $(a)$ follows from the volume of the truncated octahedron as a function of its edge length \cite{coxeter1973regular}, and $(b)$ follows from (\ref{Dl}).

For two big cells connecting from their square faces, we have:
\begin{align}
&D_u=R\sqrt{6(m_1^2+m_2^2+m_3^2)-4(m_1m_2+n_1m_3+m_2m_3)},\\
&R_B=\frac{D_u}{2\sqrt{2}}.
\end{align}

Then, the integer frequency reuse will be:
\begin{equation} \label{Mnumber}
q=\frac{V_B}{V_S}=\big(\frac{D_u}{2\sqrt{2}R}\big)^3{=}\sqrt{\frac{\big[3(m_1^2+m_2^2+m_3^2)-2(m_1m_2+m_1m_3+m_2m_3)\big]^3}{64}}. 
\end{equation}
Since the number of cells per cluster represents the frequency reuse factor is a positive integer, $(n_1,n_2,n_3)$ and $(m_1,m_2,m_3)$ must generate an integer in (\ref{Nnumber}) and (\ref{Mnumber}).
\end{proof}	
 Theorem \ref{Th2} can be used to determine the feasible integer frequency reuse factors in the considered 3D network. In addition, while performing frequency planning, the 3D locations of co-channel cells (i.e., drone-BSs) can be identified. As an example, the frequency reuse of one is obtained by considering $(n_1,n_2,n_3)=(1,0,0)$, and $(m_1,m_2,m_3)=(1,1,0)$. In fact, $q=1$ corresponds to a worst-case scenario in which all the drone-BSs will interfere with each other. In this case, the locations of co-channel interfering drone-BSs corresponding to a reference cell with an edge length $R$ and center $(0,0,0)$ are the columns of the following matrix:
\begin{equation}
\boldsymbol{H}=\sqrt{2} R\begin{pmatrix}  
	1 & 1 & -1 & 1 & 1 & -1 & -1 & -1 & 1 & -1 & 2 & 0 & 0 & -2 & 0 & 0\\
	-1 & 1 & 1& 1& 1& -1&  1& -1& -1& -1& 0& 2& 0&  0& -2& 0\\
	1 & -1 & 1& -1& 1& -1& -1& 1& -1& 1&  0& 0& 2&  0& 0& -2
%	1 & -1 & 1\\
%	1 & 1 & -1\\
%    -1 & 1 & 1\\
%    1 & 1 & -1\\
%    1 & 1 & 1\\
%    -1 & -1 & -1\\
%    -1 & 1 & -1\\
%    -1 & -1 & 1\\
%    1 & -1 & -1\\
%    -1 & -1 & 1\\ 
%      2 & 0 & 0\\
%      0 & 2 & 0\\
%      0 & 0 & 2\\
%     -2 & 0 & 0\\
%     0 & -2 & 0\\
%     0 & 0 & -2
\end{pmatrix},
\end{equation}
where each column of matrix $\boldsymbol{H}$ represents a 3D location of one co-channel drone-BS.

In summary, our approach for 3D deployment and frequency planning of drone-BSs can proceed as follows. We deploy the first drone-BS as a reference cell in a specified space of interest. Then, using our truncated octahedron model with parameter $R$, we use Theorem \ref{T1} to find the locations of other drone-BSs with respect to the reference cell. In this case, each drone-BS is located at the center of one truncated octahedron cell.  This results in a truncated octahedron tessellation that covers a given space without any gap or overlap. For frequency planning, we use Theorem \ref{Th2} to find the feasible frequency reuse factors. Then, for any given frequency reuse factor, we determine the sets of co-channel cells in the network. This, in turn, enables us to compute the SINR and transmission latency (which is used in our optimization problem in (\ref{OPT1})) at any location in the 3D space.

 \begin{figure}[!t]
 	\begin{center}
 		\includegraphics[width=9cm]{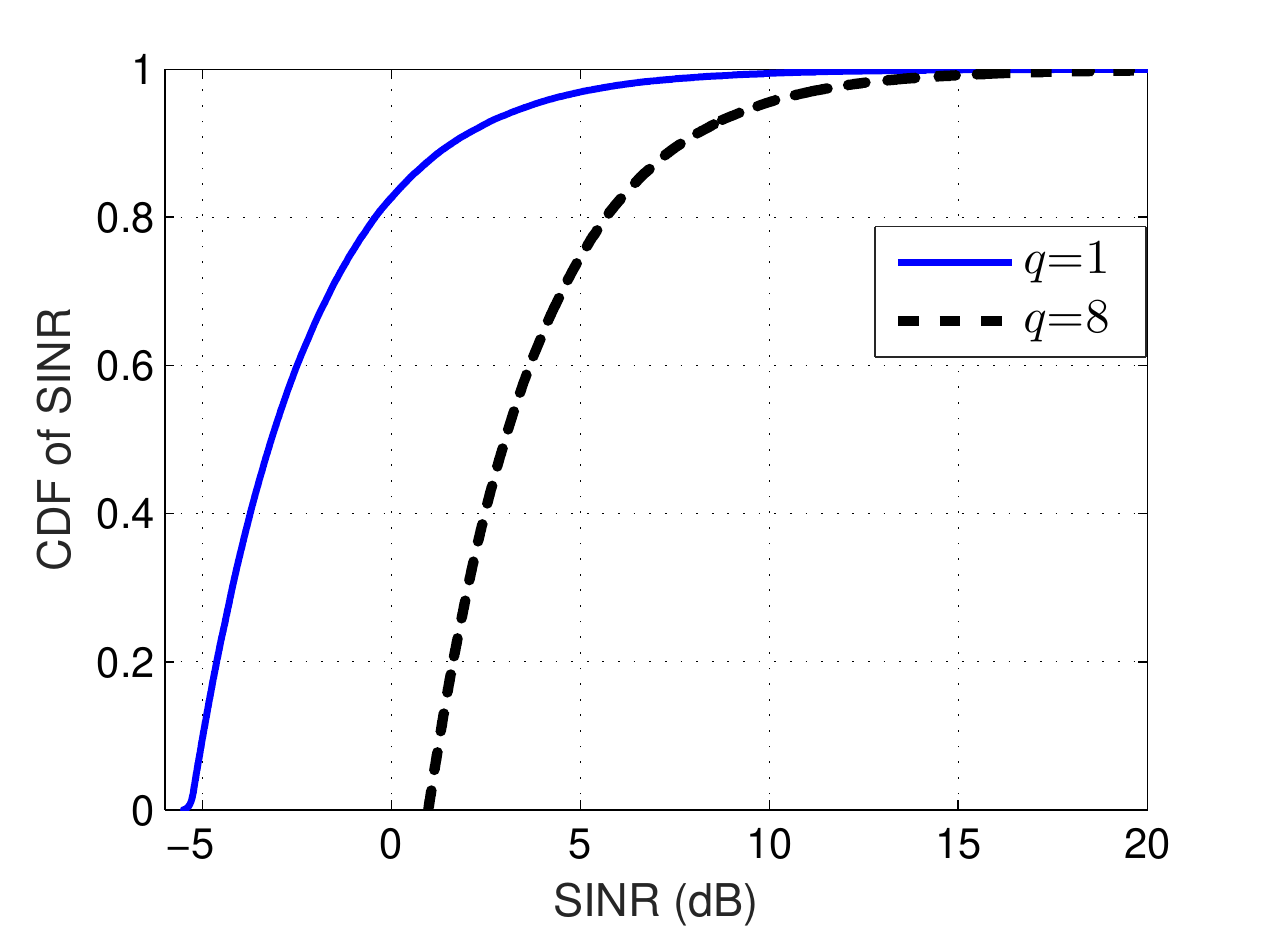}
 		\vspace{-0.4cm}
 		\caption{CDF of drone-UEs' SINR in a 3D cell for two different frequency reuse factors. \vspace{-.82cm}}
 		\label{fig:CDF}
 	\end{center}\vspace{-0.2cm}	
 \end{figure}

To show the impact of the frequency reuse factor on the SINR of drone-UEs, in Fig. \ref{fig:CDF}, we plot the cumulative distribution function (CDF) of drone-UEs' SINR in a 3D cell with a $R=400$\,m. As we expect, drone-UEs experience higher SINR for a higher frequency reuse factor (i.e., $q$). However, a case with a frequency reuse factor 8 requires eight time more bandwidth compared to the case of frequency reuse 1.

\section{Estimation of the Spatial Distribution of Drone-UEs} \label{sec:EstDist}

Since drone-UEs cannot continuously report their locations due to excessive overhead costs, we need to design a machine learning based mechanism for estimating the locations of drone-UEs using sparse information. Therefore, we assume that each drone-UE is able to report its location at each $T$ seconds. Then, using that, we estimate the spatial distribution of drone-UEs which remains valid for the next $T$ seconds. %In our method, locations of some drone-UEs can be sufficient for the desired estimation accuracy. However increasing number of samples will  increase the estimation accuracy. 
We should note that, during $T$ seconds, the location of each drone-UE is changing due to its mobility. However, the distribution of drone-UEs is fixed so that we can use our estimation for the period of $T$ seconds. To this end, we develop a nonparametric model for $f(x,y,z)$ using a kernel density estimation (KDE) \cite{bishop2007pattern}. In case of parametric density estimation methods, if one uses a poor assumption for the density model, it results in a poor estimation performance. However, nonparametric methods are not sensitive to such poor assumptions.

 The distribution of drone-UEs changes with time. Nevertheless, since we assume that this distribution is fixed within an interval of $T$ seconds, we sample the location of each drone-UE every $T$ seconds, and use it to estimate $f(x,y,z)$. This reduces overhead compared to the case in which the system knows the location of drone-UEs at every time instant. We consider some small regions $\mathcal{R}$ where each drone lies in with probability $p$. Hence, the number of drone-UEs in this region $K$ follows a binomial distribution, i.e.,
 \begin{equation}
 \text{Pr}(K)=\frac{L!}{(L-K)!K!}p^K(1-p)^{L-K},
 \end{equation}
 For a binomial distribution, we know that the mean is $\mathds{E}(\frac{K}{L})=p$. Thus, we can write:
 \begin{equation}\label{eq:ML:lim_infty}
 \lim_{L\to \infty} \frac{K}{L}=p.
 \end{equation}
 Therefore, for a large $L$, we can assume $K=Lp$. Since  $\mathcal{R}$ is a small region, we can assume that
 $f(x,y,z), \forall (x,y,z) \in \mathcal{R}$ is constant, and hence: 
 \begin{equation}\label{eq:ML:int_smallregion}
 p=\int_{\mathcal{R}} f(x,y,z) \textrm{d}x \textrm{d}y \textrm{d}z=f(x,y,z) \,\mathcal{V}_{\mathcal{R}},
 \end{equation}
 where $\mathcal{V}_{\mathcal{R}}$ is the volume of region $\mathcal{R}$. Combining (\ref{eq:ML:lim_infty}) and (\ref{eq:ML:int_smallregion}), we can write:
 \begin{equation}
 f(x,y,z)= \frac{K}{L \mathcal{V}_{\mathcal{R}}}.
 \end{equation}
 If we define a small region $\mathcal{R}$ as a cube:
 \begin{equation}\label{eq:ML:cube_def}
 \mathcal{C}(\frac{x}{h_x},\frac{y}{h_y},\frac{z}{h_z})=
 \begin{cases}
 1, &\max\{|\frac{x}{h_x}|,|\frac{y}{h_y}|,|\frac{z}{h_z}|\}\leq 1/2,\\
 0, &\text{otherwise},
 \end{cases}
 \end{equation}
 then, we can write the total number of users inside this cube as:
 \begin{equation}
 K=\sum_{i=1}^L \mathcal{C}\left(\frac{x-x_i}{h},\frac{y-y_i}{h},\frac{z-z_i}{h}\right)=L h_xh_yh_z f(x,y,z).
 \end{equation}
 Since the volume of the cube in (\ref{eq:ML:cube_def}) is $h_x\cdot h_y\cdot h_z$, we can write the density function  as:
 \begin{equation}\label{eq:ML:cube_sum}
 f(x,y,z)=\frac{1}{L}\sum_{i=1}^L \frac{1}{h_xh_yh_z}\mathcal{C}\left(\frac{x-x_i}{h_x},\frac{y-y_i}{h_y},\frac{z-z_i}{h_z}\right),
 \end{equation}
 which can be interpreted as $L$ cubes with the volume $h_x\cdot h_y\cdot h_z$ centered at each data point. Also, $h_x$, $h_y$, and $h_z$ are the widths of the kernel in dimensions $x$, $y$, and $z$, respectively. To remove the discontinuity of cubes in the space, we use Gaussian kernels \cite{kernel2002Elgammal}. If we approximate each cube in (\ref{eq:ML:cube_sum}) with a Gaussian kernel, we have:
 \begin{equation}\label{eq:ML:build_model}
 \hat{f}(x,y,z)=\frac{1}{L}\sum_{i=1}^L \frac{1}{\sqrt{(2\pi)^3h_xh_yh_z}}e^{-\left(\frac{(x-x_i)^2}{h_x}+\frac{(y-y_i)^2}{h_y}+\frac{(z-z_i)^2}{h_z}\right)}.
 \end{equation}
 $\hat{f}(x,y,z)$ is not equal to $f(x,y,z)$, for two reasons. First, $L$ is a finite number, and second, the Gaussian kernel is an approximation of the cube in (\ref{eq:ML:cube_def}). However, we will see that this estimation has small errors even when the value of $L$ is not large. 
 We assume that $x$, $y$, and $z$ are uncorrelated, and hence, all the off-diagonal elements of the covariance matrix are zero. Here, the parameters $h_x$, $h_y$, and $h_z$ have a major effect on the accuracy of the estimation and need to be estimated. 
 The criteria for accuracy of kernel density estimation is the mean integrated squared error (MISE)  and for our problem, it is given by:
 \begin{equation}
 e=\mathds{E}\bigg[ \int_{-\infty}^{\infty}\int_{-\infty}^{\infty}\int_{-\infty}^{\infty} \left(\hat{f}(x,y,z;h_x,h_y,h_z)-f(x,y,z)\right)^2 \textrm{d}x \textrm{d}y \textrm{d}z\bigg].
 \end{equation}
 Since the MISE is not a mathematically tractable expression except in special cases,  we have to use approximation methods for approximating it. To this end, we first write MISE as:
 \begin{equation}
 \mathds{E}\bigg[ \int_{-\infty}^{\infty}\int_{-\infty}^{\infty}\int_{-\infty}^{\infty} \hat{f}^2(x,y,z;h_x,h_y,h_z)+f^2(x,y,z)-2\hat{f}(x,y,z;h_x,h_y,h_z)f(x,y,z) \textrm{d}x \textrm{d}y \textrm{d}z\bigg],
 \end{equation}
 where $h_x$, $h_y$, and $h_z$ are solutions to the following minimization problem:
 \begin{equation}\label{eq:ML:argmin-h}
 [h_x,h_y,h_z]=\text{arg} \min \mathds{E}\bigg[ \int_{-\infty}^{\infty}\int_{-\infty}^{\infty}\int_{-\infty}^{\infty} \hat{f}^2(x,y,z)-2\hat{f}(x,y,z;h_x,h_y,h_z)f(x,y,z) \textrm{d}x \textrm{d}y \textrm{d}z\bigg],
 \end{equation}
 where $f^2(x,y,z)$ has been omitted since it is a constant in the minimization problem. We can approximate (\ref{eq:ML:argmin-h}) using leave-one-out cross-validation (LOOCV) methods. 
 To this end, we first build a model for $\hat f(x,y,z;h)$ using the locations of all drone-UEs except one \cite{zambom2012reviewKernel}. Then, we find the log-likelihood for the remaining drone-UEs' locations using the current model. We repeat this operation and take an average with $L$ log-likelihood values, i.e.,
 \begin{align}
 \mathcal{L}(h_x,h_y,h_z)&=\frac{1}{L}\sum_{j=1}^L \hat{f}_{-j}(X_j,Y_j,Z_j;h_x,h_y,h_z)\\
 &=\frac{1}{L}\sum_{\substack{i=1\\ i\neq j}}^L \frac{1}{\sqrt{(2\pi)^3h_xh_yh_z}}e^{-\left(\frac{(X_j-x_i)^2}{h_x}+\frac{(Y_j-y_i)^2}{h_y}+\frac{(Z_j-z_i)^2}{h_z}\right)}.
 \end{align}
 It can be shown \cite{Turlach_bandwidthselection,kernel1995multivariate} that:
 \begin{equation}
 \mathds{E} \bigg[\hat{f}(x,y,z;h_z,h_y,h_z)\bigg]=\mathcal{L}(h_x,h_y,h_z),
 \end{equation}
 and since
 \begin{equation}
 \mathds{E}\bigg[ \int_{-\infty}^{\infty}\int_{-\infty}^{\infty}\int_{-\infty}^{\infty} \hat{f}(x,y,z;h_x,h_y,h_z)f(x,y,z) \textrm{d}x \textrm{d}y \textrm{d}z\bigg]=\mathds{E} \bigg[\hat{f}(x,y,z;h_z,h_y,h_z)\bigg],
 \end{equation}
 we can find $h_x$, $h_y$, and $h_z$ by a cross-validation method  as:
 \begin{equation}
 h_x,h_y,h_z=\text{arg} \min \mathds{E}\bigg[ \int_{-\infty}^{\infty}\int_{-\infty}^{\infty}\int_{-\infty}^{\infty} \hat{f}^2(x,y,z)\textrm{d}x \textrm{d}y \textrm{d}z-2\mathcal{L}(h_x,h_y,h_z)\bigg].
 \end{equation}
 Hence, we can say that  $-\mathcal{L}(h_x,h_y,h_z)$ is a biased estimator of MISE. Therefore, it can predict the location of the minimum  MISE, and using that, we can find the optimal $h_x$, $h_y$, and $h_z$.
  Algorithm \ref{ML:Alg1} summarizes the estimation of $f(x,y,z)$ using location of drone-UEs in each $T$ seconds.
 \begin{algorithm}[t]
 	\caption{drone-UEs' distribution estimation algorithm} \label{KernalAlg}
 	\label{ML:Alg1}
 	\begin{footnotesize}
 	\begin{algorithmic}
 		\renewcommand{\algorithmicrequire}{\textbf{Input:}}
 		\renewcommand{\algorithmicensure}{\textbf{Output:}}
 		\Require drone-UEs location $(X_1,Y_1,Z_1)\cdots,(X_L,Y_L,Z_L)$
 		\Ensure  $\hat{f}(x,y,z)$
 		
 		\State \textbf{Initialize:} $\mathcal{H}\gets$ set of candidate for  $\{h_x,h_y,h_z\}$, $\mathcal{L}(h_{\text{best}})\gets \infty$

 		\For {${h_x,h_y,h_z} \in \mathcal{H}$}
 		\For {$j=1,\cdots,L$}
 		\State Build a model using (\ref{eq:ML:build_model}) with $X_i, i\in \{1,\cdots,L\}, i\neq j$
 		\State sum$\gets$ sum$+\frac{1}{2}\log h_x+\frac{1}{2}\log h_y+\frac{1}{2}\log h_z+\left(\frac{(X_j-x_i)^2}{h_x}+\frac{(Y_j-y_i)^2}{h_y}+\frac{(Z_j-z_i)^2}{h_z}\right)+\frac{3}{2}\log(2\pi)$
 		\EndFor
 		\State $\mathcal{L}(h_x,h_y,h_z) \gets \frac{1}{L}$sum
 		\If {$\mathcal{L}(h_x,h_y,h_z)\leq \mathcal{L}(h_{\text{best}})$}
 		\State $h_{\text{best}}\gets h_x,h_y,h_z$
 		\EndIf
 		\EndFor
 		\State $h_x,h_y,h_z \gets h_{\text{best}}$
 		\State return $\hat{f}(x,y,z;h_x,h_y,h_z)$ in (\ref{eq:ML:build_model}) as drone-UEs PDF
 		
 	\end{algorithmic}
 	\end{footnotesize}
 \end{algorithm}

 \begin{figure}[!t]
 	\begin{center}
 		\vspace{-0.1cm}
 		\includegraphics[width=8.5cm]{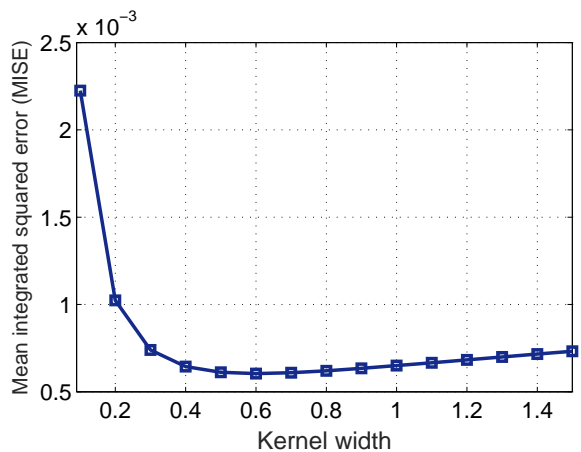}
 		\vspace{-0.4cm}
 			\caption{ٍMISE for symmetric kernel widths ($h_x=h_y=h_z=h$). } \vspace{-.4cm}
 		\label{fig:MISE}
 	\end{center}	\vspace{-0.3cm}
 \end{figure}

% \begin{figure}[!t]
% 	\centering
% 	\includegraphics[scale=.6,trim={0.6cm 0 0 0}]{Figures/MISE.eps}
% 	\caption{ٍMISE for symmetric kernels ($h_x=h_y=h_z=h$). }
% 	\label{fig:MISE}
% \end{figure}

 Fig. \ref{fig:MISE} shows the MISE in case of symmetric kernels ($h_x=h_y=h_z=h$) for $30$ drone-UEs for different values of $h$. As we can see, our algorithm can potentially minimize the MISE to a value of $7.9554\times 10^{-04}$. Fig. \ref{fig:LOOCV} shows the negative log-likelihood function. As we can see from Figs. \ref{fig:MISE} and \ref{fig:LOOCV}, $-\mathcal{L}(h_x,h_y,h_z)$ is a biased estimator of MISE, and hence, we can use $-\mathcal{L}(h_x,h_y,h_z)$ to find the optimal $h_x,h_y,h_z$. Fig. \ref{fig:LOOCV} shows that, by means of LOOCV method, the MISE for our PDF estimation is $5.7221\times 10^{-04}$ which is close to the MISE lower bound that is $7.9554\times 10^{-04}$.
 
 In summary, our approach for estimation of drone-UE spatial distribution is as follows.  We collect the location of drone-UEs at each $T$ seconds. Then, we estimate the distribution of drone-UEs to use it for 3D cell association during the next $T$ seconds. We adopt an accuracy metric for our density estimation and use it to find width of the kernels. We showed that our approach is able to estimate the spatial distribution of drone-UEs with a near optimal accuracy.
% \begin{figure}[!t]
% 	\centering
% 	\includegraphics[scale=.6,trim={0.6cm 0 0 0}]{Figures/Negative_LogLik.eps}
% 	\caption{ٍLOOCV method for finding optimum $h$. }
% 	\label{fig:LOOCV}
% \end{figure}

\begin{figure}[!t]
	\begin{center}
		\vspace{-0.1cm}
		\includegraphics[width=8.5cm]{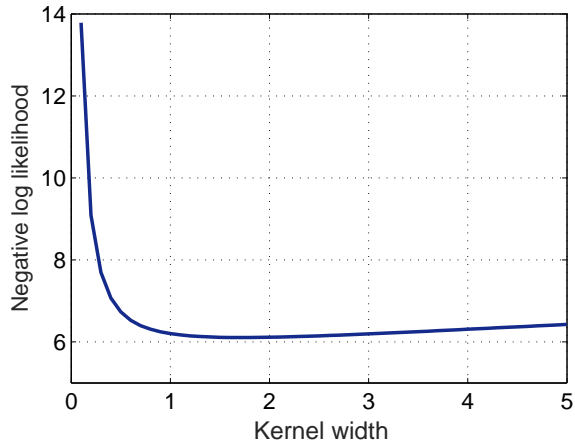}
		\vspace{-0.4cm}
	\caption{ٍLOOCV method for finding an optimal kernel width ($h)$. } \vspace{-.4cm}
		\label{fig:LOOCV}
	\end{center}	\vspace{-0.3cm}
\end{figure}

 \section{Optimal 3D Cell Association for Minimum Latency} 
 In Sections III and IV, we determined the locations of drone-BSs and the spatial distribution of drone-UEs. Here, we use this information (i.e., drone-BSs' locations and drone-UEs' distribution) to explicitly formulate our latency-optimal 3D cell association problem. 
   \begin{align} \label{OPT2}
   \mathop {\min }\limits_{{\mathcal{V}_1,...,\mathcal{V}_N}} &\mathlarger{\sum}\limits_{n = 1}^N \Bigg[ {\bigintssss_{{\mathcal{V}_n}}{\frac{\beta K_n}{B_n\log_2\big(1+\gamma_n(x,y,z)\big)}  \hat{f}(x,y,z)}} \textrm{d}x\textrm{d}y\textrm{d}z+\frac{\beta K_n}{C_n}+g_n(\beta K_n)\Bigg], \\[5pt]
   \textrm{s.t.}\,\,\,\, &K_n=L\int_{\mathcal{V}_n}{ \hat{f}(x,y,z) \textrm{d}x\textrm{d}y\textrm{d}z},\\
   &{\mathcal{V}_l} \cap {\mathcal{V}_m} = \emptyset ,\,\,\,\forall l \ne m \in \mathcal{N}, \label{VOL1}\\
   &\bigcup\limits_{n \in \mathcal{N}} {{\mathcal{V}_n}}  = \mathcal{V}, \label{VOL2}
   \end{align}
 where $\gamma_n(x,y,z)$ is the downlink SINR for a drone-UE located at $(x,y,z)$ which is served by drone-BS $n$. Considering a practical bounded path loss model \cite{liu2017effect} for air-to-air communications, the SINR can be given by:
\begin{align}
&\gamma_n(x,y,z)=\frac{\eta \textcolor{black}{\kappa_n(x,y,z)} P_n  [1+d_n(x,y,z)]^{-\alpha}}{\mathlarger{\sum}\limits_{u \in {\mathcal{I}_\textrm{int}}} {\eta \kappa_u(x,y,z){P_u}[1+d_u(x,y,z)]^{-\alpha}}+ N_oB_n},\\
&d_n(x,y,z)=\sqrt{(x-x_n)^2+(y-y_n)^2+(z-z_n)^2},\\
&d_u(x,y,z)=\sqrt{(x-x_u)^2+(y-y_u)^2+(z-z_u)^2},\,\, u\in \mathcal{I}_\textrm{int},
\end{align}
\textcolor{black}{where $\kappa_n(x,y,z)$ is a channel gain factor between a drone-UE, located at $(x,y,z)$, and drone-BS $n$. $\kappa_n(x,y,z)$ depends on the environment, and the locations of the drone-UE and drone-BS. $\kappa_n(x,y,z)=1$ corresponds to a LoS air-to-air communication, while $0<\kappa_n(x,y,z)<1$ can capture the impact of NLoS conditions.} $\alpha$ is the path loss exponent, $N_o$ is the noise power spectral density, $\eta$ is the path loss constant, and $(x_n,y_n,z_n)$ is the 3D location of drone-BS $n$. $d_n(x,y,z)$ and $d_u(x,y,z)$ are, respectively, the distance of drone-BSs $n$ and $u$ with a drone-UE located at $(x,y,z)$. Also, $\mathcal{I}_\textrm{int}$ is the set of co-channel interfering drone-BSs that operate over the same frequency band as drone-BS $n$.
 
 Solving (\ref{OPT2}) is challenging since  the optimization variables $\mathcal{V}_n$, $ \forall n \in \mathcal{N}$, are  continuous 3D association spaces which are mutually dependent. Furthermore, the fact that the size and shape of these 3D association spaces are unknown, exacerbates the complexity. In addition, the objective function in (\ref{OPT2}) does not have a closed-form expression thus making the problem intractable. Consequently, employing traditional optimization techniques (e.g., convex optimization) are not sufficient to  solve  (\ref{OPT2}). Here, we tackle our 3D space association by exploiting optimal transport theory. In particular, first, we prove the existence of an optimal solution to (\ref{OPT2}) and, then, we completely characterize the solution space. We note that, compared to our previous work in \cite{MozaffariFlightTime}, this work is different in terms of the system model, the 3D cell association optimization problem, as well as the solution.

 Optimal transport theory is a mathematical tool that is used to find an optimal
 mapping between two arbitrary probability measures \cite{villani2003}. More specifically, in a semi-discrete optimal
 transport problem, a continuous probability density function must be mapped to a discrete
 probability measure. In such a semi-discrete case, the optimal transport map will optimally
 partition the continuous distribution and assign each partition to one point in the discrete
 probability measure (which, in our case, is the discrete set of drone-BSs). 
 
 Our cell association problem can be modeled as a semi-discrete optimal transport problem in
 which the source measure (drone-UEs' distribution) is continuous while the destination (distribution
 of drone-BSs) is discrete. Then, the optimal 3D cell partitions are obtained by optimally mapping the
drone-UEs to drone-BSs.
 \begin{lemma} \normalfont
 	The optimization problem in (\ref{OPT2}) admits an optimal solution for any semi-continuous function $g_n(.)$, $n\in \mathcal{N}$.
 \end{lemma}
 \begin{proof}
 	Consider $K_n=L\int_{\mathcal{V}_n}{\hat{f}(x,y,z) \textrm{d}x\textrm{d}y\textrm{d}z}$ and the following simplex:
 		\begin{equation}
 		\hspace{-0.1cm}	S \hspace{-0.07cm}= \hspace{-0.07cm}\left\{\hspace{-0.07cm} {\boldsymbol{K} = \left( {{K_1},{K_2},...,{K_{N}}} \right) \hspace{-0.06cm}\in {\mathds{R}^{N}};\sum\limits_{n = 1}^{N} \hspace{-0.05cm}{{K_n}= L}, {K_n} \ge 0, \forall n\in \mathcal{N} \hspace{-0.02cm}} \right\}\hspace{-0.02cm}.
 		\end{equation}
 Given any vector $\boldsymbol{K}$, the optimization problem in (\ref{OPT2}) can be represented by:
 \begin{align} \label{OPTLemma}
 \mathop {\min }\limits_{{\mathcal{V}_1,...,\mathcal{V}_N}} &\mathlarger{\sum}\limits_{n = 1}^N {\bigintssss_{{\mathcal{V}_n}}{c(\boldsymbol{v},\boldsymbol{s}_n) \hat{f}(\boldsymbol{v})}} \textrm{d}\boldsymbol{v}, \\[5pt]
 \textrm{s.t.}\,\,\,\, &\int_{\mathcal{V}_n}{ \hat{f}(\boldsymbol{v}) \textrm{d}\boldsymbol{v}}=\frac{K_n}{L},\\
 &{\mathcal{V}_l} \cap {\mathcal{V}_m} = \emptyset ,\,\,\,\forall l \ne m \in \mathcal{N},\,\,\, \bigcup\limits_{n \in \mathcal{N}} {{\mathcal{V}_n}}  = \mathcal{V}, \label{VOL2v2}
 \end{align}
where $\boldsymbol{s_n}$ is the location of drone-BS $n$, $\boldsymbol{v}=(x,y,z)$, and $c(\boldsymbol{v},\boldsymbol{s}_n)= \frac{\beta K_n}{B_n\log_2\big(1+\gamma_n(x,y,z)\big)}+\frac{L}{K_n}(\frac{\beta K_n}{C_n}+g_n(\beta K_n))$.
 
This optimization problem is equivalent to the following semi-discrete optimal transport problem:
	\begin{equation} \label{Opt5}
	\mathop {\min }\limits_T \int_\mathcal{V} {c\left( {\boldsymbol{v},\boldsymbol{s}} \right)}  \hat{f}(\boldsymbol{v})\textrm{d}\boldsymbol{v}, \,\, \boldsymbol{s}=T(\boldsymbol{v}),
	\end{equation}
	where $\boldsymbol{s}$ is the location of a drone-BS, and $T(.)$ is the transport map which is related to 3D cell partition $\mathcal{V}_n$ by:\vspace{-0.1cm}
	\begin{equation}
	\left\{T(\boldsymbol{v}) = \sum\limits_{n=1}^ {N} {\boldsymbol{s}_n{\mathds{1}_{{\mathcal{V}_n}}}(\boldsymbol{v})}; \int_{{\mathcal{V}_n}} \hat{f}(\boldsymbol{v})\textrm{d}\boldsymbol{v}=\frac{K_n}{L}\right\}.
	\end{equation}
 Considering the fact that  for any semi-discrete optimal transport problem with a lower semi-continuous cost function an optimal transport map exists \cite{villani2003, Crippa}, (\ref{OPTLemma}) admits an optimal solution for any $\boldsymbol{K}\in S$. Also, since $S$ is a simplex (which is a non-empty and compact set), problem (\ref{OPT2}) admits an optimal solution over the entire $S$. This proves Lemma 1. 
 \end{proof}
 
 Next, given the existence of the optimal solution, we characterize the solution.
 
 \begin{theorem} \label{Theorem3DCells} \normalfont
 	The optimal 3D cell association for drone-BS $l$, that leads to a minimum average latency in (\ref{OPT2}), is given by:
 	\begin{align} \label{OPTcell}
 	\mathcal{V}_l^*=\Big\{(x,y,z)\big|&\alpha_l + \frac{K_l}{L} h_l(x,y,z)+\frac{\beta}{C_l}+g'_l(\beta K_l) \nonumber
 	\\
 	&\le \alpha_m + \frac{K_m}{L} h_m(x,y,z)+\frac{\beta}{C_m}+g'_m(\beta K_m),\, \forall l\neq m \Big\},
 	\end{align}	
 	where $h_l(x,y,z)\triangleq\frac{\beta }{B_l\log_2\big(1+\gamma_l(x,y,z)\big)}$, and $\alpha_l\triangleq\int_{{{\mathcal{V}}_l}} {h_l(x,y,x) \hat{f}(x,y,z)\textrm{d}x\textrm{d}y\textrm{d}z}$.
\end{theorem}
\begin{proof}
	See Appendix $A$.
\end{proof}

 Using Theorem \ref{Theorem3DCells}, we can determine the optimal 3D cell partitions associated with each drone-BS that ensure the minimum average latency for drone-UEs. From (\ref{OPTcell}), we can see how the optimal 3D association depends on various network's parameters such as the distribution of drone-UEs, locations of drone-BSs, backhaul data rate, load of the network, and the computational speed. Based on these parameters, Theorem \ref{Theorem3DCells} is utilized to optimally partition a specified space and determine a minimum latency 3D cell association scheme. In this case, to minimize the average latency, a drone-BS with a faster backhaul link and computational capabilities, or higher bandwidth and transmit power will serve more drone-UEs.

 \textcolor{black}{ To solve (\ref{OPTcell}), we propose the iterative algorithm shown in Algorithm \ref{OTAlgor}. This algorithm, based on  \cite{Crippa}, converges to the
  optimal solution within a reasonable number of iterations. The complexity of this iterative approach mainly depends on computing the numerical integration in Step 6 of Algorithm 2. A practical approach to compute this integration is to use a pixel-based integration as given in \cite{merigot}. This approach is practical to implement as its complexity grows linearly with the size of the considered 3D space $\mathcal{V}$.} 
  Algorithm \ref{OTAlgor} for solving (\ref{OPTcell}) that finds the optimal 3D cell partitions proceeds as follows. The inputs are the 3D spatial distribution of drone UEs, number of drone-UEs, load, locations of the drone-BSs, computation time function, and the number of iterations, $Q$. In Algorithm \ref{OTAlgor}, $t$  represents the iteration number. First, we generate initial 3D cell partitions $\mathcal{V}_l^{(t)}$, and set $\psi _l^{(t)}(x,y,z)=0$, $\forall l \in \mathcal{N}$, with $\psi _l^{(t)}(x,y,z)$ being a pre-defined parameter which is used to update the cell partitions. Next, we update  $\psi _l^{(t+1)}(x,y,z)$, and compute $K_l$ in step \ref{1}. In step \ref{3}, we update the partitions based on (\ref{OPTcell}). Finally, we obtain the optimal 3D cell partitions and associations, at the end of the iteration.

 In summary, our approach for deployment and latency-optimal cell association in the proposed 3D cellular network is as follows. First, using the proposed truncated octahedron approach, and Theorems 1 and 2 in Section III, we determine the locations of drone-BSs as well as the co-channel cells. Then, in Section IV, we estimate the spatial distribution of drone-UEs using kernel method presented in Algorithm 1. Finally, based on the determined locations of drone-BSs and the drone-UEs' distribution, we use Algorithm 2 to derive the optimal 3D cell association for which the average total latency of serving drone-UEs is minimized.
 
 \begin{algorithm}[!t]
 	\begin{small}
 		\caption{Iterative algorithm for finding the optimal 3D cell association.}
 		\label{OTAlgor}
 		\begin{algorithmic}[1] 
 			\State \textbf{Inputs}: $ \hat{f}(x,y,z)$, $\beta$, $Q$, $L$, \textrm{Locations of drone-BSs}, $C_l$, $g_l(.)$, $\forall l \in\mathcal{N}$.
 			\State \textbf{Outputs:} $\mathcal{V}_l^*$,  $\forall l \in \mathcal{N}$. 
 			\State {Set $t=1$, generate an initial cell partitions $\mathcal{V}_l^{(t)}$, and set $\psi _l^{(t)}(x,y,z)=0$, $\forall l \in\mathcal{N}$}.
 			
 			\While {$t<Q$} 
 			
 			\State {Compute $\psi _l^{(t + 1)}(x,y,z) = \left\{ \begin{array}{l}
 				\left[ {1 - 1/t} \right]\psi _l^{(t)}(x,y,z),\hspace{1.8cm} \textrm{if}\,\,(x,y,z) \in {\mathcal{V}_l}^{(t)},\\
 				1 - \left[{1 - 1/t} \right]\left( {1 - \psi _l^{(t)}(x,y,z)}\right),\,\, \textrm{otherwise}.
 				\end{array} \right.$  } \label{phi}

 			\State {Compute ${K_l} = \bigintssss_{\mathcal{V}} {\left( {1 - \psi _l^{(t + 1)}(x,y,z)} \right)\hat{f}(x,y,z)} \textrm{d}x\textrm{d}y\textrm{d}z$, $\forall l \in \mathcal{N}$}. \label{1}

 			\State {$t \to t+1$}.\label{2}
 			
 			\State {Update cell partitions using (\ref{OPTcell})}.\label{3}
 			
 			%	\State {Update $\tau_i^{(t)}$ using (\ref{tau})}.\label{5}

 			%	\State {Repeat steps \ref{phi} to \ref{3}.}
 			
 			\EndWhile
 			\State 	{$\mathcal{V}_l^*= \mathcal{V}_l^{(t)}$,} \label{A*}

 		\end{algorithmic}
 	\end{small} 
 \end{algorithm}%\vspace{-0.3cm}

\section{Simulation Results and Analysis}

For our simulations, we consider a cubic space of size $3\,\text{km}\times 3\,\text{km}\time3\times3\,\text{km}$ in which 18 drone-BSs are deployed based on the proposed truncated octahedron approach to serve drone-UEs. We determine the locations of drone-BSs by using (\ref{POS}) with parameters $a\in\{-1,0,1\}, b\in\{-1,0,1\}$, $c\in\{0,1\}$, and $R=400$\,m. We randomly generate a sample (i.e., a realization of a continuous distribution) of drone-UEs' locations based on a three-dimensional truncated Gaussian distribution with a specified mean and variance values. These locations' samples are then used to estimate the spatial distribution of drone-UEs using Algorithm \ref{KernalAlg}. For the computation time, we consider a quadratic function of  data size (i.e., load on each drone-BS), but our approach can accommodate any other arbitrary function. Here, the computation time for drone-BS $n$ is $g_n(\beta K_n)=\frac{(\beta K_n)^2}{\omega_n}$, with $\omega_n$ being the processing speed of drone-BS  $n$. Unless states otherwise, we use the simulation parameters listed in Table I. We compare our proposed 3D cell association with the classical SINR-based association (i.e., weighted Voronoi diagram) baseline. All statistical results are averaged over a large number of independent runs.

\begin{table}[!t]
	\normalsize
	\begin{center}
		%\centering
		\caption{\small Simulation parameters.}
		\vspace{-0.1cm}
		\label{TableP}
		\resizebox{13cm}{!}{
			\begin{tabular}{|c|c|c|}
				\hline
				\textbf{Parameter} & \textbf{Description} & \textbf{Value} \\ \hline \hline
				$f_c$	&     Carrier frequency     &      2\,GHz     \\ \hline 
				$P_n$	&    Drone-BS transmit power     &     0.5\,\textrm{W}    \\ \hline
				
				$N_o$	&     Noise power spectral density    &        -170\,dBm/Hz  \\ \hline
				%	$\gamma$	&     SINR threshold      &   5\,dB \\ \hline
				$L$	&     Number of drone-UEs     &   200 \\ \hline

				$B_n$	&    Bandwidth for each drone-BS      &   10\,MHz \\ \hline
				
				$\alpha$	&    Path loss exponent      &   2 \\ \hline
				
				$\eta$	&     Path loss constant     &   $1.42\times 10^{-4}$ \\ \hline
				
				$\beta$	&     Packet size for drone-UE     &   10\,kb \\ \hline
				
				$q$	&     Frequency reuse factor     &   1 \\ \hline
				
				$C_n$	&     Backhaul rate for drone-BS $n$    &   $(100+n)$ Mb/s \\ \hline
				
			     $\omega_n$	&    Computation constant (i.e., speed) for each drone-BS    &   $10^2$ Tb/s \\ \hline
				
				$\mu_x, \mu_y, \mu_z$& Mean of the truncated Gaussian distribution in $x$, $y$, and $z$ directions & 1000\,m,\,1000\,m,\,1000\,m\\ \hline
				
					$\sigma_x, \sigma_y, \sigma_z$& Standard deviation of the distribution in $x$, $y$, and $z$ directions & 600\,m,\,600\,m,\,600\,m\\ \hline
				
		$\kappa_n$	&    Channel gain factor    &   1 \\ \hline
			\end{tabular}}
			
		\end{center}\vspace{-0.5cm}
	\end{table}

Fig.\,\ref{LatencyUeser} shows the average total latency as a function of the number of drone-UEs for the proposed 3D cell association and the SINR-based association schemes. As we can see from this figure, the total latency increases by increasing the number of drone-UEs. A higher number of drone-UEs leads to a higher network congestion which, in turn, increases transmission time, backhaul latency, and computation time. Fig.\,\ref{LatencyUeser} shows that, when the number of drone-UEs increases from 200 to 300, the total latency increases by 56\% and 42\% for the SINR-based association and the proposed approach. Moreover, we can see that our proposed approach significantly reduces the latency compared to the SINR association case. This is due to the fact that, in our approach, besides SINR, the impact of congestion on the transmission, backhaul, and computational latencies is also taken into account. The proposed approach avoids creating highly congested 3D cell partitions that can cause excessive latency. \textcolor{black}{From Fig.\,\ref{LatencyUeser}, we can see that our approach yields, on the average, 43.9\% reduction in the average total latency compared to the SINR-based association.}

%Fig.\,\ref{LatencyUeser} shows the average total latency as a function of the number of drone-UEs for the proposed 3D cell association and the SINR-based association schemes. As we can see from this figure, the total latency increases by increasing the number of drone-UEs. A higher number of drone-UEs leads to a higher network congestion which, in turn, increases transmission time, backhaul latency, and computation time. Fig.\,\ref{LatencyUeser} shows that, when the number of drone-UEs increases from 300 to 500, the total latency increases by 68\% and 59\% for the SINR-based association and the proposed approach. Moreover, we can see that our proposed approach significantly reduces the latency compared to the SINR association case. This is due to the fact that, in our approach, besides SINR, the impact of congestion on the transmission, backhaul and computation latencies is also taken into account. The proposed approach avoids creating highly congested 3D cell partitions that can cause excessive latency. From Fig.\,\ref{LatencyUeser}, we can see that our approach yields around 48\% reduction in the average total latency compared to the SINR-based association.  

\begin{figure}[!t]
	\begin{center}
		\vspace{-0.1cm}
		\includegraphics[width=9.5cm]{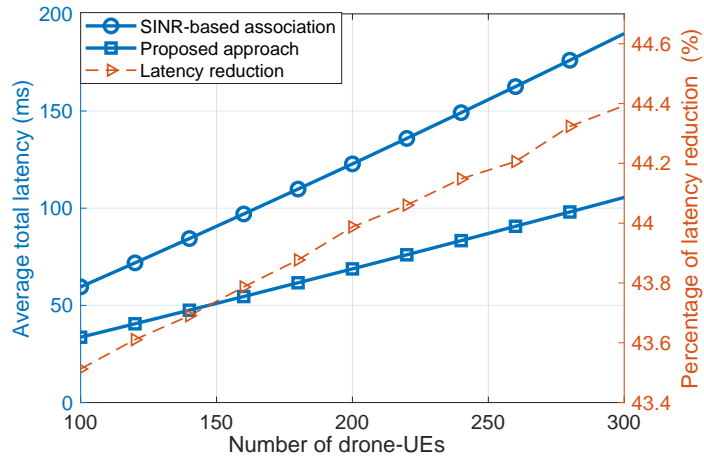}
		\vspace{-0.4cm}
		\caption{ Average total latency vs. number of drone-UEs.} \vspace{-.4cm}
		\label{LatencyUeser}
	\end{center}	\vspace{-0.5cm}
\end{figure}

Fig.\,\ref{LatencyBW} shows how the latency can be reduced by increasing the transmission bandwidth. By using more bandwidth, the transmission rate increases and, hence, the transmission latency decreases. Fig.\,\ref{LatencyBW} also reveals that our approach significantly enhances spectral efficiency compared to the SINR-based association. In essence, compared to the SINR case, the proposed approach requires less transmission bandwidth in order to meet a certain latency requirement. For instance, as we can see from Fig.\,\ref{LatencyBW}, to ensure a 70\,ms maximum total latency, our approach requires 57\% less bandwidth than the SINR-based association scheme. Another observation from Fig.\,\ref{LatencyBW} is that the rate of latency reduction decreases as the bandwidth increases. %For example, in this case, we do not achieve considerable latency improvement while increasing the bandwidth from 8\,MHz to 10\,MHz. 
This is because in large bandwidth scenarios, the transmission latency can be smaller than the computation and backhaul latencies. Thus, the impact of decreasing the transmission latency on the total latency is relatively minor. \vspace{-0.1cm}

%Fig.\,\ref{LatencyBW} shows how the latency can be reduced by increasing the transmission bandwidth. By using more bandwidth, the transmission rate increases and, hence, the transmission latency decreases. Fig.\,\ref{LatencyBW} also reveals that our approach significantly enhances spectral efficiency compared to the SINR-based association. In other words, compared to the SINR case, the proposed approach requires lower amount of transmission bandwidth in order to meet a certain latency requirement. For instance, as we can see from Fig.\,\ref{LatencyBW}, to ensure a 10 minutes maximum total latency, our approach requires 61\% less bandwidth than the SINR-based association scheme. Another observation from Fig.\,\ref{LatencyBW} is that the rate of latency reduction decreases as the bandwidth increases. For example, in this case, we do not achieve considerable latency improvement while increasing the bandwidth from 8\,MHz to 10\,MHz. This is because in large bandwidth scenarios, the transmission latency can be smaller than the computation and backhaul latencies. Therefore, the impact of decreasing the transmission latency on the total latency will be relatively minor.

\begin{figure}[!t]
	\begin{center}
		\vspace{-0.1cm}
		\includegraphics[width=9.5cm]{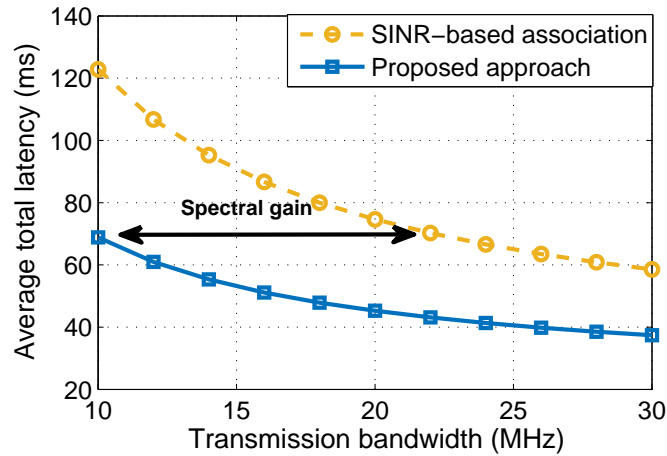}
		\vspace{-0.4cm}
		\caption{ Average total latency vs. transmission bandwidth.} \vspace{-.4cm}
		\label{LatencyBW}
	\end{center}	\vspace{-0.5cm}
\end{figure}

Fig.\, \ref{All_Latency} shows the impact of drone-UEs' load on the transmission, computation, and backhaul latencies. As expected, these three types of latency increase when the load of drone-UEs increases. Nevertheless, the rate of increase is different for different types of latency. For instance, in Fig.\,\ref{All_Latency}, the increase rate of the transmission latency is higher than that of computational latency and backhaul latency. The impact of load on each type of latency depends on two factors: 1) the function that directly relates the load to the latency, and 2) the 3D cell partitions which are related to load by (\ref{OPTcell}). In fact, while varying load, the cell partitions and different component of latency dynamically change such that the total latency is minimized.

\begin{figure}[!t]
	\begin{center}
		\vspace{-0.1cm}
		\includegraphics[width=9.0cm]{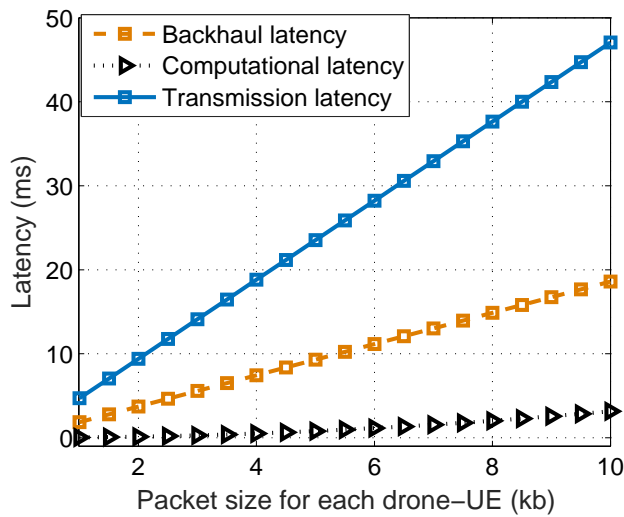}
		\vspace{-0.4cm}
		\caption{Transmission, backhaul, and computation latency vs. load of each drone-UE in the proposed approach.} \vspace{-.4cm}
		\label{All_Latency}
	\end{center}	\vspace{-0.5cm}
\end{figure}  

{\color{black}
In Fig.\,\ref{Latency_T}, we evaluate the impact of sampling time, $T$ on the latency (which depends accuracy of  3D cell partitions). To this end, we consider a time varying distribution for drone-UEs  whose mean changes by $T$. In Fig.\,\ref{Latency_T}, we consider a three-dimensional truncated Gaussian 
distribution whose mean value increases by $\nu T$, with $\nu$ being a rate of distribution change. In Fig.\,\ref{Latency_T}, we show how the latency increases by increasing $T$, for $\nu=0.1$. Note that, while the accuracy of distribution estimation increases by reducing $T$, this results in a higher complexity and overhead in the considered network.   In Fig.\,\ref{Latency_T}, we show the additional
latency that can be caused by an error in estimating the drone-UEs' distribution. 
Clearly, the total latency significantly depends on the 3D cell partitions which themselves are
a function the drone-UEs' distribution. Therefore, an estimation error in the drone-UEs' distribution
leads to a deviation from the optimality of the cell partitions. Consequently, such estimation error will increase the latency. Hence, by decreasing $T$, the network can obtain a more accurate distribution estimation and, hence, a lower latency. For instance, as we can see from Fig.\,\ref{Latency_T}, the latency decreases by 6\% when decreasing the sampling time from 20 minutes to 10 minutes, for $\nu=10$\,m/min. 

%Although decreasing $T$ does affect the accuracy of the prediction and hence increasing the latency for the users. However decreasing T will increase the overhead per time and subsequently increasing the delay. However there might be a possible optimum $T$ for each mobility pattern of UAV-UEs. 

\begin{figure}[!t]
	\begin{center}
		\vspace{-0.1cm}
		\includegraphics[width=9.5cm]{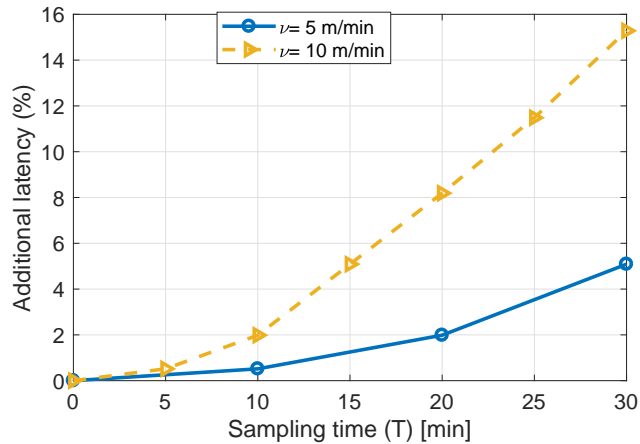}
		\vspace{-0.4cm}
		\caption{ \textcolor{black}{Additional latency vs. sampling time for distribution estimation ($T$)}} \vspace{-.4cm}
		\label{Latency_T}
	\end{center}	\vspace{-0.5cm}
\end{figure} 

}

\begin{comment}
In Fig.\,\ref{ErrorLatency}, we show how any error in estimating the spatial distribution of drone-UEs can affect the system performance in terms of latency.  As an example, in this figure, we show the additional latency that can be caused by an error in estimating the mean value of the drone-UEs' distribution. Clearly, the total latency significantly depends on the 3D cell partitions which themselves are a function drone-UEs' distribution. Therefore, an estimation error of drone-UEs' distribution leads to a deviation from the optimality of cell partitions. Consequently, such estimation error will increase the latency. For instance, from Fig.\,\ref{ErrorLatency}, we can observe around 27\% increase in the average total latency when increasing the error up to 400\,m. Hence, while estimating the distribution of drone-UEs,  it is crucial to adopt high-accuracy models, as we did in Section IV. 

\begin{figure}[!t]
	\begin{center}
		\vspace{-0.1cm}
		\includegraphics[width=9cm]{./Figures/ExtraLatency_MeanError.eps}
		\vspace{-0.4cm}
		\caption{Additional latency due to estimation error in mean of drone-UEs' distribution.} \vspace{-.5cm}
		\label{ErrorLatency}
	\end{center}	\vspace{-0.5cm}
\end{figure}  

\end{comment}

\begin{figure}[!t]
	\begin{center}
		\vspace{-0.1cm}
		\includegraphics[width=9cm]{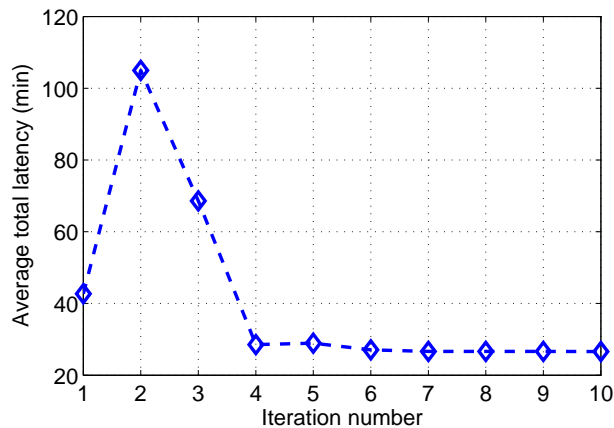}
		\vspace{-0.4cm}
		\caption{ \textcolor{black}{Convergence of Algorithm \ref{OTAlgor}.}} \vspace{-.4cm}
		\label{Covergence}
	\end{center}	\vspace{-0.1cm}
\end{figure}

\textcolor{black}{Finally, in Fig.\, \ref{Covergence}, we show the convergence of Algorithm \ref{OTAlgor} that is used to find the optimal 3D cell association by iteratively solving (\ref{OPTcell}). As we can see from this figure, Algorithm \ref{OTAlgor} converges within 6 iterations.}

\section{Conclusion}\vspace{-0.10cm}
In this paper, we have introduced a novel framework for cell association and deployment in 3D cellular networks with drone-BSs and drone-UEs. We have proposed a tractable method for the 3D deployment of drone-BSs and solved the problem of cell association with the goal of minimizing the latency of drone users.  For deployment, we have determined the drone-BSs' locations based on a truncated octahedron structure and  derived the feasible frequency reuse factor in the considered 3D network. For latency-minimal cell association, first, we have estimated the spatial distribution of the drone-UEs using the kernel density estimation method. Then, using the estimated distribution of drone-UEs and the location of drone-BSs, we have derived the optimal cell association of drone-UEs using optimal transport theory such that the latency for drone-UEs is minimized. Our results have shown that the proposed approach significantly reduces the latency of drone-UEs compared to the classical SINR-based association. Furthermore, the proposed latency-optimal cell association improves the spectral efficiency of the 3D drone-enabled wireless networks. \vspace{-0.2cm}

\section*{Appendix}\vspace{-0.2cm}
\subsection{Proof of Theorem \ref{Theorem3DCells}}\vspace{-0.2cm}
 %\begin{proof}
 	In Lemma 1, we proved the existence of the optimal 3D cell partitions ${\mathcal{V}}_n$, $n\in \mathcal{N}$. Now, consider two 3D partitions ${\mathcal{V}}_l$ and ${\mathcal{V}}_m$, and a point $\boldsymbol{v}_o=(x_o,y_o,z_o)\in {\mathcal{V}}_l$. Also, let $B_\epsilon(\boldsymbol{v}_o)$ be a ball with a center $\boldsymbol{v}_o$ and radius $\epsilon >0$. \textcolor{black}{Now, we generate the following new 3D partitions ${{\mathord{\buildrel{\lower3pt\hbox{$\scriptscriptstyle\frown$}}\over {\mathcal{V}}} }_n}$ (which are variants of the optimal partitions):} \vspace{-0.1cm}
 	%\begin{small}
 	\begin{equation}
 	\left\{ \begin{array}{l}
 	\widehat{\mathcal{V}}_l = {\mathcal{V}}_l\backslash {B_\varepsilon }({\boldsymbol{v}_o}),\\
 	\widehat{\mathcal{V}}_m = {\mathcal{V}}_m \cup {B_\varepsilon }({\boldsymbol{v}_o}),\\
 	\widehat{\mathcal{V}}_n = {\mathcal{V}}_n,\,\,\,\,n \ne l,m.
 	\end{array} \right.
 	\end{equation}
 	%\end{small}
 	Let us define $p_1(K_n)\triangleq{K_n}$, $p_2(K_n)\triangleq\frac{\beta K_n}{C_n}$, ${K_\varepsilon } = L\int_{{B_\varepsilon }({\boldsymbol{v}_o})} { \hat{f}(x,y,z)\textrm{d}x\textrm{d}y\textrm{d}z}$, and $\widehat{K}_n = L\int_{\widehat{\mathcal{V}}_n} { \hat{f}(x,y,z)\textrm{d}x\textrm{d}y\textrm{d}z}$. Considering the optimality of ${\mathcal{V}}_n$, $n\in \mathcal{N}$, we have:
 	%\begin{small}
 	\begin{align}
 	&\hspace{0.1cm}\sum\limits_{n \in \mathcal{N}} {\int_{{{\mathcal{V}}_n}} { {{p_1}\left( {{K_n}} \right)h_n(x,y,z)}  \hat{f}(x,y,z)\textrm{d}x\textrm{d}y\textrm{d}z} }+p_2(K_n)+g_n(\beta K_n) \nonumber \\
 	& {\mathop  \le \limits^{(a)} }\sum\limits_{n \in \mathcal{N}} {\int_{\widehat{\mathcal{V}}_n} { {{p_1}\left( \widehat{K}_n \right)h_n(x,y,z)}  \hat{f}(x,y,z)\textrm{d}x\textrm{d}y\textrm{d}z} } +p_2(\widehat{K}_n)+g_n(\beta \widehat{K}_n). \label{SUM}
 	\end{align}
 	Canceling out the common terms in (\ref{SUM}) leads to: 
 	\begin{align}
 	\hspace{0.1cm}&\int_{{{\mathcal{V}}_l}} { {{p_1}\left( {{K_l}} \right)h_l(x,y,z) }  \hat{f}(x,y,z)\textrm{d}x\textrm{d}y\textrm{d}z} +p_2(K_l)+g_l(\beta K_l)\nonumber \\
 	&+ \int_{{{\mathcal{V}}_m}} { {{p_1}\left( {{K_m}} \right)h_m(x,y,z)}  \hat{f}(x,y,z)\textrm{d}x\textrm{d}y\textrm{d}z}+p_2(K_m)+g_m(\beta K_m)\nonumber\\
 	& \le \int_{{{\mathcal{V}}_m} \cup {B_\varepsilon }({\boldsymbol{v}_o})} { {{p_1}\left( {{K_m+K_{\epsilon}}} \right)h_m(x,y,z) } \hat{f}(x,y,z)\textrm{d}x\textrm{d}y\textrm{d}z} +p_2(K_m)+g_m(\beta (K_m+K_{\epsilon})) \nonumber\\ &+\int_{{{\mathcal{V}}_l}\backslash {B_\varepsilon }({\boldsymbol{v}_o})} { {{p_1}\left( {{K_l-K_{\epsilon}}} \right)h_l(x,y,z) }  \hat{f}(x,y,z)\textrm{d}x\textrm{d}y\textrm{d}z} +p_2(K_l-K_{\epsilon})+g_l(\beta (K_l-K_{\epsilon})),\\
 	&\int_{{{\mathcal{V}}_l}} { ({{p_1}\left( {{K_l}} \right)-{p_1}\left( K_l-{{K_{\epsilon}}} \right))h_l(x,y,z) } \hat{f}(x,y,z)\textrm{d}x\textrm{d}y\textrm{d}z} +p_2(K_l)-p_2(K_l-K_{\epsilon})\nonumber\\
 	&+g_l(\beta K_l)-g_l(\beta ( K_l-K_{\epsilon}))\nonumber +\int_{{B_\varepsilon }({\boldsymbol{v}_o})} { {{p_1}\left( {{K_l-K_{\epsilon}}} \right)h_l(x,y,z) }  \hat{f}(x,y,z)\textrm{d}x\textrm{d}y\textrm{d}z}\nonumber\\
 	&\le \int_{{{\mathcal{V}}_m}} { ({{p_1}\left( {{K_m+{{K_{\epsilon}}}}} \right)-{p_1}\left( K_m \right))h_l(x,y,z) }  \hat{f}(x,y,z)\textrm{d}x\textrm{d}y\textrm{d}z} +p_2(K_m+K_{\epsilon})-p_2(K_m)\nonumber\\
 	& +g_m(\beta(K_m+K_{\epsilon}))-g_m(\beta K_m)+ \int_{{B_\varepsilon }({\boldsymbol{v}_o})} { {{p_1}\left( {{K_m} + {K_\varepsilon }} \right)h_m(x,y,z) }  \hat{f}(x,y,z)\textrm{d}x\textrm{d}y\textrm{d}z},\label{ineq}
 	\end{align}
 	where $(a)$ comes from the fact that ${\mathcal{V}}_n$, $\forall n \in \mathcal{N}$ are optimal 3D partitions and, thus, any variation of such optimal partitions, shown by $\widehat{\mathcal{V}}_n$, does not lead to a better solution.

 	Note that, ${K_\epsilon } = L\int_{{B_\epsilon }({\boldsymbol{v}_o})} { \hat{f}(x,y,z)\textrm{d}x\textrm{d}y\textrm{d}z}$. Now, we multiply both sides of the inequality in (\ref{ineq}) by $\frac{1}{K_{\epsilon}}$ and take the limit when $\epsilon \to 0$. Then, we use the following equalities: 
 	% \begin{small}
 	\begin{align}
 	&\mathop {\lim }\limits_{\varepsilon  \to 0} {K_\epsilon } = 0,\label{C1}\\
 	&\mathop {\lim }\limits_{{K_\epsilon } \to 0} \frac{{{p_1}({K_l}) - {p_1}({K_l} - {K_\epsilon })}}{{{K_\epsilon }}} = {p'_1}({K_l}),\label{C2}\\
 	&\mathop {\lim }\limits_{{K_\epsilon } \to 0} \frac{{{p_1}({K_m} + {K_\epsilon }) - {p_1}({K_m})}}{{{K_\epsilon }}} = {p'_1}({K_m}),\label{C3}
 	\end{align}
 	%\end{small}
 	\begin{footnotesize}
 		\begin{equation} 
 		\mathop {\lim }\limits_{{K_\epsilon } \to 0} \frac{ {\bigintssss_{{B_\epsilon }({\boldsymbol{v}_o})} { {{p_1}\left( {{K_l} - {K_\epsilon }} \right)h_l(x,y,z)} \hat{f}(x,y,z)\textrm{d}x\textrm{d}y\textrm{d}z}}}{K_\epsilon}=\mathop {\lim }\limits_{{K_\epsilon } \to 0} \frac{ {{p_1}\left( {{K_l}} \right)h_l(\boldsymbol{v}_o)} {\bigintssss_{{B_\epsilon }({\boldsymbol{v}_o})} {  \hat{f}(x,y,z)\textrm{d}x\textrm{d}y\textrm{d}z}}}{K_\epsilon}=\frac{{p_1}\left( {{K_l}} \right)h_l(\boldsymbol{v}_o)}{L},\nonumber\label{C4}
 		\end{equation}		
 	\end{footnotesize}		
 	\begin{footnotesize}
 		\begin{equation}
 		\mathop {\lim }\limits_{{K_\epsilon } \to 0} \frac{ {\bigintssss_{{B_\varepsilon }({\boldsymbol{v}_o})} { {{p_1}\left( {{K_m} + {K_\epsilon }} \right)h_m(x,y,z)}  \hat{f}(x,y,z)\textrm{d}x\textrm{d}y\textrm{d}z}}}{K_\epsilon}=\mathop {\lim }\limits_{{K_\epsilon } \to 0} \frac{ {{p_1}\left( {{K_m}} \right)h_m(\boldsymbol{v}_o)} {\bigintssss_{{B_\epsilon }({\boldsymbol{v}_o})} {   \hat{f}(x,y,z)\textrm{d}x\textrm{d}y\textrm{d}z}}}{K_\epsilon}=\frac{{p_1}\left( {{K_m}} \right)h_m(\boldsymbol{v}_o)}{L}.\nonumber\label{C5}
 		\end{equation}		
 	\end{footnotesize}	
 	Finally, using (\ref{C1})-(\ref{C5}), we obtain:
 	%\begin{small}
 	\begin{align}
 	&{p'_1\left( {{K_l}} \right)\int_{{{\mathcal{V}}_l}} h_l(x,y,z) \hat{f}(x,y,z)\textrm{d}x\textrm{d}y\textrm{d}z}+\frac{1}{L} {p_1}\left( {{K_l}} \right)h_l({\boldsymbol{v}_o})+p'_2(K_l)+g'_l(\beta K_l)\nonumber\\
 	&\le {p'_1\left( {{K_m}} \right) \int_{{{\mathcal{V}}_m}} h_m(x,y,z)\hat{f}(x,y,z)\textrm{d}x\textrm{d}y\textrm{d}z}  + \frac{1}{L}{p_1}\left( {{K_m}} \right)h_m(\boldsymbol{v}_o)+p'_2(K_m)+g'_m(\beta K_m). \label{eq9}
 	\end{align}
 	%\end{small}
 	Note that, in $p'_1(K_l)$, the derivative is taken with respect to a single variable which is written as ${p'_1}({K_l}) = {\left. {\frac{{d{p_1}(t)}}{{dt}}} \right|_{t = {K_l}}}$.  
 	
 	We can further proceed to derive a tractable expression for (\ref{eq9}):
 	
 	Given ${p_1}(K_l) = K_l$, we can compute ${p'_1}({K_l}) = 1$, then, using ${K_l} \hspace{-0.05cm}=\hspace{-0.05cm} \int_{{{\mathcal{V}}_l}} {\hat{f}(x,y,z)\textrm{d}x\textrm{d}y\textrm{d}z}$ leads~to:
 	%\begin{small}
 	\begin{align}
 	\alpha_l + \frac{1}{L}K_l h_l({\boldsymbol{v}_o})+\frac{\beta}{C_l}+g'_l(\beta K_l)
 	\le \alpha_m + \frac{1}{L}K_m h_m({\boldsymbol{v}_o})+\frac{\beta}{C_m}+g'_m(\beta K_m),
 	\end{align}
 	As a result, each optimal 3D cell association can be represented by:
 	\begin{align}
 	\mathcal{V}_l^*=\Big\{(x,y,z)\big|&\alpha_l + \frac{K_l}{L} h_l(x,y,z)+\frac{\beta}{C_l}+g'_l(\beta K_l) \nonumber
 	\\
 	&\le \alpha_m + \frac{K_m}{L} h_m(x,y,z)+\frac{\beta}{C_m}+g'_m(\beta K_m),\, \forall l\neq m \Big\},
 	\end{align}
 	%\end{small}  
 	which completes the proof of Theorem \ref{Theorem3DCells}.
\def\baselinestretch{1.2692}
\bibliographystyle{IEEEtran}

\bibliography{references}
%\vspace{-0.4cm}
% that's all folks
\end{document}